\documentclass[preprint,12pt]{elsarticle}
\usepackage[hmargin=1.3in,vmargin=1.3in]{geometry}
\usepackage{bbding}
\usepackage{mathrsfs}
\usepackage{amsfonts}
\usepackage{multirow}
\usepackage{amsfonts,amssymb,amsmath,amsthm,bm}
\usepackage[T1]{fontenc}
\usepackage[utf8]{inputenc}
\usepackage[colorlinks,
            linkcolor=blue,
            anchorcolor=blue,
            citecolor=blue]{hyperref}
\usepackage{graphicx}
\usepackage{booktabs}
\usepackage{longtable}
\usepackage{caption}
\usepackage{romanbar}
\usepackage{tikz}
\usepackage{pstricks}
\usepackage{graphicx} % 用于包含图形
\usepackage[justification=centering]{caption} % 用于控制caption的位置
\usepackage{paralist}
\usepackage{array}

\newtheorem{theorem}{Theorem}

\newtheorem{definition}{Definition}
\newtheorem{remark}{Remark}
\newtheorem{corollary}{Corollary}
\newtheorem{lemma}{Lemma}
\newtheorem{proposition}{Proposition}

\journal{Finite Fields and Their Applications }
\begin{document}
\begin{frontmatter}
\title{Three Classes of Twisted Gabidulin Codes with Different Twists%$^\dag$
\footnote{$^\dag$This research is supported by the National Key Research and Development Program of China (Grant No. 2022YFA1005000), the National Natural Science Foundation of China (Grant Nos. 12141108, 62371259, 12411540221, 62201322), the Fundamental Research Funds for the Central Universities of China (Nankai University), the Nankai Zhide Foundation.}}
\author[1]{Ran Li\corref{cor}}
\ead{ran111212@163.com}
\author[1]{Fang-Wei Fu}
\ead{fwfu@nankai.edu.cn}
\author[2,3,4]{Weijun Fang}
\ead{fwj@sdu.edu.cn}
{\cortext[cor]{Corresponding author.}}
%\tnotetext[label1]
%{Submitted for possible publication.}
\address[1]{Chern Institute of Mathematics and LPMC, Nankai University, Tianjin, 300071, China}
\address[2]{State Key Laboratory of Cryptography and Digital Economy Security, Shandong University, Qingdao, 266237, China\\}
\address[3]{Key Laboratory of Cryptologic Technology and Information Security, Ministry of Education, Shandong University, Qingdao, 266237, China}
\address[4]{School of Cyber Science and Technology, Shandong University, Qingdao, 266237, China}

\begin{abstract}
Twisted Gabidulin codes are an extension of Gabidulin codes and have recently attracted great attention. In this paper, we study three classes of twisted Gabidulin codes with different twists. Moreover, we establish necessary and sufficient conditions for them to be maximum rank distance (MRD) codes, determine the conditions under which they are not MRD codes, and construct several classes of MRD codes via twisted Gabidulin codes. In addition, considering these codes in the Hamming metric, we provide necessary and sufficient conditions for them to be maximum distance separable (MDS), almost MDS, or near MDS. Finally, we investigate the covering radii and deep holes of twisted Gabidulin codes.
\end{abstract}
\end{frontmatter}
\section{Introduction}
Rank metric codes have received widespread attention in coding theory since they were introduced in the seminal work of Delsarte \cite{dels}. With their error-correction capability, these codes offer distinct advantages in network coding \cite{network}, space-time coding \cite{space}, and cryptographic applications \cite{crypto}. Many important properties of rank metric codes, including the Singleton-like bound, were independently studied by Delsarte \cite{dels}, Gabidulin \cite{1985G} and Roth \cite{roth}. A rank metric code achieving the Singleton-like bound is termed a maximum rank distance (MRD) code.

In 1985, Gabidulin \cite{1985G} proposed a class of rank metric codes, known as Gabidulin codes, which are MRD codes. Gabidulin codes are evaluation codes of linearized polynomials \cite{ore}, defined using the Frobenius automorphism of a finite field extension $\mathbb{F}_{q^{m}}/\mathbb{F}_{q}$. In \cite{kshe,rothrm}, the Frobenius automorphism in Gabidulin codes was generalized to arbitrary automorphism, and generalized Gabidulin codes were proposed. In recent years, twisted codes have been widely studied. In \cite{shee}, Sheekey modified the evaluation polynomials of a Gabidulin code by adding an extra monomial, resulting in twisted Gabidulin codes. Choosing the parameters of the twisted Gabidulin codes in an appropriate way can ensure that the codes are MRD.  Using the same idea for generalizing Gabidulin codes, Sheekey constructed generalized twisted Gabidulin codes in \cite{shee}, which were later extensively studied in \cite{GTG,GTG2,GTG3}. In \cite{otal}, a large family of MRD codes, called additive generalized twisted Gabidulin codes, was introduced. This family contains all the aforementioned linear MRD codes as subfamilies, as well as new additive MRD codes. Additionally, the idea of twisted Gabidulin codes was adapted to the Hamming metric and further generalized in \cite{trs1} and \cite{strs}, resulting in the so-called twisted Reed-Solomon codes. Later, along this line of research, a large number of meaningful results on twisted Reed-Solomon codes have been obtained in \cite{trs,1,2,3,4}.

In coding theory, covering radii and deep holes are important parameters for linear codes. The covering radius of a code is the least positive integer $\rho$ such that the union of the spheres of radius $\rho$ centered at each codeword is equal to the full ambient space. This fundamental coding theoretical parameter has been widely studied for codes with respect to the Hamming metric \cite{dp1,h1,h3,h4,h5,h7}, but only a few papers on the subject have appeared in the literature on rank metric codes \cite{procr,end,st,fu,packpro}. For the rank metric, if the rank distance from a vector to the code $\mathcal{C}$ achieves the covering radius of the code, the vector is called a deep hole of the code $\mathcal{C}$.

 In this paper, we study twisted Gabidulin codes with one twist, two twists and $\ell$ twists. Specifically, we derive necessary and sufficient conditions for them to be MRD codes, determine when they are not MRD codes, and construct several classes of MRD twisted Gabidulin codes. Additionally, we study these codes in the Hamming metric, and we establish necessary and sufficient conditions for them to be MDS, AMDS, or NMDS. Finally, we discuss their covering radii and obtain some deep holes for certain twisted Gabidulin codes.

The rest of the paper is organized as follows. In Section 2, we provide the necessary notations, definitions and preliminary results.
In Section 3, we establish the necessary and sufficient conditions for twisted Gabidulin codes with one twist to be MRD, MDS, AMDS, or NMDS. Moreover, we determine the conditions under which they are not MRD. In Section 4, we present the necessary and sufficient conditions for twisted Gabidulin codes with two twists to be MRD, MDS, AMDS, or NMDS. Furthermore, we determine when they are not MRD and construct a class of MRD codes. In Section 5, we derive the necessary and sufficient conditions for twisted Gabidulin codes with $\ell$ twists to be MRD, MDS, AMDS, or NMDS. We also establish the conditions under which they are not MRD and construct two classes of MRD codes. In Section 6, we study the covering radii and deep holes of twisted Gabidulin codes. In Section 7, we conclude the paper.

\section{Preliminaries}
In this section, we introduce some notations, definitions and results which will be used in subsequent sections.
\subsection{Notations}
Let $\mathbb{F}_{q}$ be a
finite field with $q$ elements and $\mathbb{F}_{q^{m}}$ be an extension field of $\mathbb{F}_{q}$ of degree $m$, where $q$ is a prime power and $m$ is a positive integer. Let $\mathbb{F}_{q^m}^{*}:=\mathbb{F}_{q^m}\backslash \{0\}$, $\mathbb{F}_{q^{m}}^{n}$ denote the $n$-dimensional vector space over $\mathbb{F}_{q^{m}}$, and $\mathbb{F}_{q^m}^{k \times n}$ denote the set of $k\times n$ matrices over  $\mathbb{F}_{q^m}$.

An  $[n, k]$  linear code  $\mathcal{C}$  with length  $n$  and dimension  $k$  over  $\mathbb{F}_{q^m}$  is a  $k$-dimensional subspace of  $\mathbb{F}_{q^m}^{n}$. If the row vectors of a matrix $G \in \mathbb{F}_{q^m}^{k \times n}$ form a basis of the code $\mathcal{C}$, then $G$ is called the generator matrix of $\mathcal{C}$. The dual
code $\mathcal{C}^{\perp}$ of $\mathcal{C}$ is the orthogonal complement space of $\mathcal{C}$ under the Euclidean inner product over $\mathbb{F}_{q^m}^{n}$.

 For a vector $\boldsymbol{x}\in \mathbb{F}_{q^m}^{n}$, the Hamming weight of $\boldsymbol{x}$, denoted by $\mathrm{wt}_{H}(\boldsymbol{x})$, is the number of non-zero components of $\boldsymbol{x}$. For two vectors $\boldsymbol{x},\boldsymbol{y}\in \mathbb{F}_{q^m}^{n}$, the Hamming distance between them is defined as $d_H(\boldsymbol{x},\boldsymbol{y}):=\mathrm{wt}_{H}(\boldsymbol{x}-\boldsymbol{y})$. Let $\mathcal{C}$ be an $[n,k]$ linear code over $\mathbb{F}_{q^{m}}$, the minimum Hamming distance of $\mathcal{C}$ is defined as $d_H(\mathcal{C}):=\mathrm{min} \left\{\mathrm{wt}_{H}(\boldsymbol{c}) : \boldsymbol{c}\in \mathcal{C}\setminus \{\boldsymbol{0}\} \right\}.$ For the linear code $\mathcal{C}$, its minimum Hamming distance satisfies the Singleton bound $d_H(\mathcal{C})\le n-k+1.$ If $d_H(\mathcal{C})=n-k+1$, then $\mathcal{C}$ is called a maximum distance separable (MDS) code. If $d_H(\mathcal{C})=n-k$, then $\mathcal{C}$ is called an almost MDS (AMDS) code. If $\mathcal{C}$ and $\mathcal{C}^{\perp}$ are both AMDS, then $\mathcal{C}$ is called a near MDS (NMDS) code.

\begin{proposition}[\cite{dp2}]\label{gmds}
Let $\mathcal{C}\subseteq \mathbb{F}_{q^{m}}^{n}$ be an $[n,k]$ linear code with a generator matrix $G$. Then $\mathcal{C}$ is an MDS code if and only if all maximal minors of $G$ are non-zero.
\end{proposition}

Throughout this paper, we fix some notations for convenience.
\begin{itemize}
\item For $\alpha\in \mathbb{F}_{q^{m}}$, let $\sigma(\alpha)=\alpha^{q}$ be the Frobenius automorphism. For any integer $i$, we have $\sigma^{i}(\alpha)=\alpha^{q^i}$.
\item Let $[i]:=q^{i}$ denote the $i$-th Frobenius power. For  $\boldsymbol{\alpha}=[\alpha_{1},\ldots,\alpha_{n}]\in \mathbb{F}_{q^m}^{n}$, we denote $\boldsymbol{\alpha}^{[i]}:=[\alpha_{1}^{[i]},\ldots,\alpha_{n}^{[i]}]$.
\item Let $\textit{\textbf{[}}n\textit{\textbf{]}}$ denote the set $\{1,\ldots,n\}.$
\item We denote the set
$$ \mathcal{V}_q(k,n):=\{V\in \mathbb{F}_{q}^{k\times n}\mid\mathrm{rk}_{\mathbb{F}_{q}}(V)=k\}.$$
\item For $\boldsymbol{\alpha}=[\alpha_{1},\ldots,\alpha_{n}]\in \mathbb{F}_{q^{m}}^{n}$, we denote the matrix
    \begin{equation}\label{mkht}
    M_{k}^{(h,t)}(\boldsymbol{\alpha}):=\left[\begin{array}{c}
\boldsymbol{\alpha} \\
\vdots  \\
\boldsymbol{\alpha}^{[h-1]}  \\
\boldsymbol{\alpha}^{[t]}  \\
\boldsymbol{\alpha}^{[h+1]} \\
\vdots  \\
\boldsymbol{\alpha}^{[k-1]}
\end{array}\right]=\left[\begin{array}{cccc}
\alpha_{1} & \alpha_{2} & \ldots & \alpha_{n} \\
\vdots & \vdots & \ddots & \vdots \\
\alpha_{1}^{q^{h-1}} & \alpha_{2}^{q^{h-1}} & \ldots & \alpha_{n}^{q^{h-1}} \\
\alpha_{1}^{q^{t}} & \alpha_{2}^{q^{t}} & \ldots & \alpha_{n}^{q^{t}} \\
\alpha_{1}^{q^{h+1}} & \alpha_{2}^{q^{h+1}} & \ldots & \alpha_{n}^{q^{h+1}} \\
\vdots & \vdots & \ddots & \vdots \\
\alpha_{1}^{q^{k-1}} & \alpha_{2}^{q^{k-1}} & \ldots & \alpha_{n}^{q^{k-1}}
\end{array}\right].
\end{equation}
\end{itemize}

\subsection{Rank metric codes}
\begin{definition}
For a vector $\boldsymbol{x}\in \mathbb{F}_{q^{m}}^{n}$, the rank weight of $\boldsymbol{x}$, denoted by $\mathrm{rk}_{\mathbb{F}_{q}}(\boldsymbol{x})$, is defined as the dimension of the linear space spanned by the components of $\boldsymbol{x}$ over $\mathbb{F}_{q}$.
\end{definition}
\begin{remark}\label{rkwt}
For a vector $\boldsymbol{x}\in \mathbb{F}_{q^{m}}^{n}$, it is clear that
 $\mathrm{rk}_{\mathbb{F}_{q}}(\boldsymbol{x})\leq \mathrm{wt}_{H}(\boldsymbol{x}).$
\end{remark}
\begin{definition}
For two vectors $\boldsymbol{x}, \boldsymbol{y}\in \mathbb{F}_{q^{m}}^{n}$, the rank distance between $\boldsymbol{x}$ and $\boldsymbol{y}$ is defined as  $d_R(\boldsymbol{x},\boldsymbol{y}):=\mathrm{rk}_{\mathbb{F}_{q}}(\boldsymbol{x}-\boldsymbol{y}).$
\end{definition}
 Throughout this paper, we consider $n\leq m$.
\begin{definition}
For an $[n,k]$ linear rank metric code $\mathcal{C}$ over $\mathbb{F}_{q^m}$ which is a $k$-dimensional subspace of $\mathbb{F}_{q^m}^n$, its minimum rank distance is defined as $d_R(\mathcal{C}):=\mathrm{min} \left\{\mathrm{rk}_{\mathbb{F}_{q}}(\boldsymbol{c}) : \boldsymbol{c}\in \mathcal{C}\setminus \{\boldsymbol{0}\} \right\}.$
\end{definition}
 Note that the rank metric codes mentioned in this paper refer to linear rank metric codes.
\begin{proposition}[\cite{dels}] For positive integers $n\le m$, let $\mathcal{C}\subseteq \mathbb{F}_{q^{m}}^{n}$ be an $[n,k]$ rank metric code, then the minimum rank distance of $\mathcal{C}$ with respect to $\mathbb{F}_{q}$ satisfies  $d_R(\mathcal{C})\le n-k+1.$
\end{proposition}
The bound given in Proposition 1 is called the Singleton-like bound. A rank metric code is called a maximum rank distance (MRD) code if it attains the Singleton-like bound. For $n\leq m$, an MRD code is MDS in the Hamming metric by Remark \ref{rkwt}.

\begin{proposition}[\cite{panduan}]\label{gmrd}
Let $G\in \mathbb{F}_{q^m}^{k\times n}$ be a generator matrix of a linear rank metric code $\mathcal{C}\subseteq \mathbb{F}_{q^m}^n$. Then $\mathcal{C}$ is an MRD if and only if
$$\mathrm{rk}_{\mathbb{F}_{q^m}}(VG^{\top})=k,\   \forall V\in \mathcal{V}_q(k,n).$$
\end{proposition}

\begin{lemma}[\cite{panduan}]\label{nmrd}
Any generator matrix $G\in \mathbb{F}_{q^m}^{k\times n}$ of an MRD code $\mathcal{C}\subseteq \mathbb{F}_{q^m}^n$ of dimension $k$ has only non-zero maximal minors.
\end{lemma}

\subsection{Linearized polynomials}
A polynomial of the form
$$f(x)=\sum_{i=0}^{s}a_i x^{q^{i}}$$  over $\mathbb{F}_{q^{m}}$ is known as a linearized polynomial \cite{ore}.
For a non-zero $f(x)=\sum_{i=0}^{s}a_i x^{q^{i}}$ over $\mathbb{F}_{q^{m}}$, its $q$-degree is  defined
by $\operatorname{deg}_q (f(x)):=
\max \left\{i: a_{i} \neq 0\right\}.$
The linearized polynomial ring, denoted by $\mathbb{F}_{q^{m}}[x;\sigma]$, is the set of linearized polynomials. The addition is done componentwise and the multiplication is defined by the non-commutative formula
$$
\left(\sum_{i=0}^{s} a_i x^{q^i}\right) \circ \left(\sum_{j=0}^{t} b_j x^{q^j}\right)=\sum_{k=0}^{s+t}\left(\sum_{i=0}^s a_i \sigma^i\left(b_{k-i}\right)\right) x^{q^k},
$$
where $b_{k-i}=0$ when  $k-i>t$ or $k-i<0$.

The following presents three useful facts.
\begin{lemma}[\cite{moore}]
Let $f(x)\in \mathbb{F}_{q^m}[x;\sigma]$ and $\mathbb{F}_{q^s}$ be the smallest extension field of $\mathbb{F}_{q^m}$ that contains all roots of $f(x)$. Then the set of all roots of $f(x)$ forms an $\mathbb{F}_{q}$-linear vector subspace of $\mathbb{F}_{q^s}$ (called the kernel  of $f(x)$).
\end{lemma}

\begin{lemma}[\cite{moore}]\label{sp1}
Let $U$ be an $\mathbb{F}_{q}$-linear subspace of $\mathbb{F}_{q^m}$. Then $\prod_{\alpha \in U}(x-\alpha)$ is an element of $\mathbb{F}_{q^m}[x;\sigma]$.
\end{lemma}

\begin{lemma}[\cite{lp2}]\label{sp2}
Let $f(x)\in \mathbb{F}_{q^m}[x;\sigma]$ such that $f(\alpha_{i})=0$ for all $i$. Then there exists $h(x)\in\mathbb{F}_{q^m}[x;\sigma]$ such that
$$f(x)=h(x) \circ \prod_{\alpha \in \left\langle\alpha_{1},\ldots,\alpha_{n}\right\rangle_{\mathbb{F}_q}}(x-\alpha).$$
\end{lemma}

\begin{definition}
Let $\boldsymbol{\alpha}=[\alpha_{1},\ldots,\alpha_{n}]\in\mathbb{F}_{q^{m}}^{n}$. A matrix $M_{k}(\boldsymbol{\alpha})$ is called a Moore matrix generated by  $\boldsymbol{\alpha}$ if it has the following form
$$
M_{k}(\boldsymbol{\alpha})=\left[\begin{array}{cccc}
\alpha_{1} & \alpha_{2} & \ldots & \alpha_{n} \\
\alpha_{1}^{[1]} & \alpha_{2}^{[1]} & \ldots & \alpha_{n}^{[1]} \\
\vdots & \vdots & \ddots & \vdots \\
\alpha_{1}^{[k-1]} & \alpha_{2}^{[k-1]} & \ldots & \alpha_{n}^{[k-1]}
\end{array}\right].
$$
\end{definition}

The following lemma states an important property of the Moore matrix.
\begin{lemma}[\cite{moore}]\label{dk}
Let $\boldsymbol{\alpha}=[\alpha_{1},\ldots,\alpha_{k}]\in\mathbb{F}_{q^{m}}^{k}$.  Then the determinant
 of the square Moore matrix $M_{k}(\boldsymbol{\alpha})$ is
$$\begin{aligned}
|M_{k}(\boldsymbol{\alpha})|=\left|\begin{array}{cccc}
\alpha_{1} & \alpha_{2} & \ldots & \alpha_{k} \\
\alpha_{1}^{[1]} & \alpha_{2}^{[1]} & \ldots & \alpha_{k}^{[1]} \\
\vdots & \vdots & \ddots & \vdots \\
\alpha_{1}^{[k-1]} & \alpha_{2}^{[k-1]} & \ldots & \alpha_{k}^{[k-1]}
\end{array}\right|=\alpha_{1} \prod_{j=1}^{k-1} \prod_{b_{1}, \ldots, b_{j} \in \mathbb{F}_{q}}\left(\alpha_{j+1}-\sum_{i=1}^{j} b_{i} \alpha_{i}\right).
\end{aligned}$$
In particular, the matrix $M_{k}(\boldsymbol{\alpha})$ is invertible if and only if $\alpha_{1},\alpha_{2},\ldots,\alpha_{k}$ are linearly independent over $\mathbb{F}_{q}$.
\end{lemma}
\begin{definition}\label{gab}
Let $n,k$ be positive integers with $k< n\leq m$ and $\boldsymbol{\alpha}\in \mathbb{F}_{q^m}^n$  with $\mathrm{rk}_{\mathbb{F}_{q}}(\boldsymbol{\alpha})=n$. The $[n,k]$ Gabidulin code $\mathcal{G}_{n,k}(\boldsymbol{\alpha})$ generated by $\boldsymbol{\alpha}$ is defined as the linear space spanned by rows of $M_{k}(\boldsymbol{\alpha})$ over $\mathbb{F}_{q^{m}}$.
\end{definition}
For $\boldsymbol{\alpha}=[\alpha_{1},\ldots,\alpha_{n}]\in \mathbb{F}_{q^{m}}^{n}$, let $\mathrm{ev}_{\boldsymbol{\alpha}}(\cdot)$ denote an evaluation map from $\mathbb{F}_{q^m}[x;\sigma]$ to $\mathbb{F}_{q^m}^n$ with $\mathrm{ev}_{\boldsymbol{\alpha}}(f(x))=[f\left(\alpha_1\right), \ldots, f\left(\alpha_n\right)].$ Then for $\mathrm{rk}_{\mathbb{F}_{q}}(\boldsymbol{\alpha})=n$, the Gabidulin code $\mathcal{G}_{n,k}(\boldsymbol{\alpha})$ can be written as
\begin{eqnarray*}
\mathcal{G}_{n,k}(\boldsymbol{\alpha})&=&\mathrm{ev}_{\boldsymbol{\alpha}}(\mathbb{F}_{q^m}[x;\sigma]_{<k})\\
&=&\left\{\left[ f\left(\alpha_{1}\right),  f\left(\alpha_{2}\right), \ldots,  f\left(\alpha_{n}\right)\right]: f(x) \in \mathbb{F}_{q^m}[x;\sigma]_{<k}\right\},\end{eqnarray*}
where $\mathbb{F}_{q^m}[x;\sigma]_{<k}$ is the set of linearized polynomials of $q$-degree at most $k-1$.

\section{Twisted Gabidulin codes with one twist}
Twisted Gabidulin codes are an extension of Gabidulin codes. In this and the following two sections, we study three classes of twisted Gabidulin codes. %in rank metric and in Hamming metric.

In this section, we consider twisted Gabidulin codes with one twist as follows.
\begin{definition}\label{l1}
Let $n, k$ be positive integers with $k<n\leq m$ and $\eta\in \mathbb{F}_{q^m}^{\ast}$. Denote $\mathcal{P}_{1}$ as the set of twisted linearized polynomials over $\mathbb{F}_{q^{m}}$ given by
\begin{equation}\label{p1}
\mathcal{P}_{1}:=\left\{f(x)=\sum_{i=0}^{k-1} f_i x^{[i]}+\eta f_{h} x^{[k+t]}: f_i \in \mathbb{F}_{q^m},0\leq i\leq k-1\right\},
\end{equation}
where $0\leq h\leq k-1$ and $0\leq t\leq n-k-1$. Let $\boldsymbol{\alpha}=[\alpha_{1},\ldots,\alpha_{n}]\in \mathbb{F}_{q^{m}}^{n}$ with $\mathrm{rk}_{\mathbb{F}_{q}}(\boldsymbol{\alpha})=n$. The twisted Gabidulin code with one twist is defined by
\begin{equation}\label{c1}
\mathcal{C}_{1}:=\mathrm{ev}_{\boldsymbol{\alpha}}\left(\mathcal{P}_{1}\right)\subseteq \mathbb{F}_{q^m}^n.
\end{equation}
\end{definition}
Starting from now on, let $\mathcal{C}_{1}$ be the twisted Gabidulin code defined as in (\ref{c1}) unless otherwise stated. We first present two useful propositions about the code $\mathcal{C}_{1}$.

\begin{proposition}
$\mathcal{C}_{1}$ is an $[n, k]$ linear code over $\mathbb{F}_{q^{m}}$ with the generator matrix
\begin{equation}\label{eqgt}
G_{1}=\left[\begin{array}{cccc}
\alpha_{1} & \alpha_{2} & \cdots & \alpha_{n} \\
\vdots &\vdots &  & \vdots \\
\alpha_{1}^{[h-1]} & \alpha_{2}^{[h-1]} & \cdots & \alpha_{n}^{[h-1]} \\
\alpha_{1}^{[h]}+\eta \alpha_{1}^{[k+t]} &\alpha_{2}^{[h]}+\eta \alpha_{2}^{[k+t]} & \cdots & \alpha_{n}^{[h]}+\eta \alpha_{n}^{[k+t]} \\
\alpha_{1}^{[h+1]} & \alpha_{2}^{[h+1]} & \cdots & \alpha_{n}^{[h+1]} \\
\vdots & \vdots & & \vdots \\
\alpha_{1}^{[k-1]} & \alpha_{2}^{[k-1]} & \cdots & \alpha_{n}^{[k-1]}
\end{array}\right].\end{equation}
\end{proposition}
\begin{proposition}
For the twisted Gabidulin code $\mathcal{C}_{1}$,
 we have $$\mathcal{C}_{1}=\left\langle\boldsymbol{\alpha}^{[h]}+\eta \boldsymbol{\alpha}^{[k+t]},\boldsymbol{\alpha}^{[s]}(s\in \{0,\ldots,k-1\}\backslash\{h\})  \right\rangle_{\mathbb{F}_{q^m}}.$$
\end{proposition}

Now we determine the values of the parameter $\eta$ for which the code $\mathcal{C}_{1}$ is MRD.

\begin{theorem}\label{mm1} $\mathcal{C}_{1}$ is MRD if and only if $$\eta\neq -\frac{\left|VM_{k}^{(h,k+t)}(\boldsymbol{\alpha})^{\top}\right|}{\left|VM_{k}(\boldsymbol{\alpha})^{\top}\right|},\ \ \ \ \forall  V\in \mathcal{V}_q(k,n),$$
where $M_{k}^{(h,k+t)}(\boldsymbol{\alpha})$ is a matrix of the form (\ref{mkht}).
\end{theorem}
\begin{proof}
Let $G_1$ given in (\ref{eqgt}) be a generator matrix of the code  $\mathcal{C}_1$. By Proposition \ref{gmrd},  $\mathcal{C}_1$  is MRD if and only if  $|V G_1^{\top}| \neq 0$  for any  $V\in \mathcal{V}_{q}(k, n)$. Given  $V\in \mathcal{V}_{q}(k, n)$, using the properties of the determinant, we obtain
$$\begin{aligned}
\left|V G_1^{\top}\right|&=\left|V\left[\begin{array}{cccc}
\boldsymbol{\alpha} \\
\vdots  \\
\boldsymbol{\alpha}^{[h-1]} \\
\boldsymbol{\alpha}^{[h]} \\
\boldsymbol{\alpha}^{[h+1]}  \\
\vdots \\
\boldsymbol{\alpha}^{[k-1]}
\end{array}\right]^{\top}\right|+\left|V\left[\begin{array}{cccc}
\boldsymbol{\alpha} \\
\vdots  \\
\boldsymbol{\alpha}^{[h-1]} \\
\eta \boldsymbol{\alpha}^{[k+t]} \\
\boldsymbol{\alpha}^{[h+1]}  \\
\vdots \\
\boldsymbol{\alpha}^{[k-1]}
\end{array}\right]^{\top}\right|=\left|V M_{k}(\boldsymbol{\alpha})^{\top}\right|+\eta\left|VM_{k}^{(h,k+t)}(\boldsymbol{\alpha})^{\top}\right|.
\end{aligned}$$
%Observe that $M_{k}(\boldsymbol{\alpha})$ is a generator matrix of the Gabidulin code $\mathcal{G}_{k}(\boldsymbol{\alpha})$. Since the Gabidulin code is MRD, we have $\left|V M_{k}(\boldsymbol{\alpha})^{\top}\right|\neq 0$.
Thus,  $\left|V G_1^{\top}\right| \neq 0$  if and only if
$\eta\left|VM_{k}^{(h,k+t)}(\boldsymbol{\alpha})^{\top}\right| \neq -\left|V M_{k}(\boldsymbol{\alpha})^{\top}\right|,$ i.e., $\eta\neq -\frac{\left|VM_{k}^{(h,k+t)}(\boldsymbol{\alpha})^{\top}\right|}{\left|VM_{k}(\boldsymbol{\alpha})^{\top}\right|}$ for any $V\in \mathcal{V}_q(k,n).$
\end{proof}
\begin{remark}\label{gneq0}
Since the Gabidulin code $\mathcal{G}_{n,k}(\boldsymbol{\alpha})$ is MRD, $\left|V M_{k}(\boldsymbol{\alpha})^{\top}\right|\neq 0$ for any $V\in \mathcal{V}_{q}(k, n)$.
\end{remark}

In \cite{furt}, Sheekey et al. constructed a class of MRD codes via twisted Gabidulin codes with $\ell$ twists. Below we give a sufficient condition for $\mathcal{C}_{1}$ to be MRD, which can be seen as a special case of \cite{furt}. Moreover, using Theorem \ref{mm1}, we provide a proof that is more direct than that of \cite{furt}.
\begin{proposition}
Let $s$ be a positive integer with $n\leq s$ such that $\mathbb{F}_{q} \subsetneq \mathbb{F}_{q^{s}} \subsetneq \mathbb{F}_{q^m}$ is a chain of subfields. Let $\alpha_{1}, \ldots, \alpha_{n}\in \mathbb{F}_{q^s}$ be linearly independent over $\mathbb{F}_{q}$ and  $\eta\in \mathbb{F}_{q^m} \backslash \mathbb{F}_{q^s}$. Then  $\mathcal{C}_{1}$  is MRD.
\end{proposition}
\begin{proof}
Since $\alpha_{1}, \ldots, \alpha_{n}\in \mathbb{F}_{q^s}$, we obtain  $\frac{\left|VM_{k}^{(h,k+t)}(\boldsymbol{\alpha})^{\top}\right|}{\left|VM_{k}(\boldsymbol{\alpha})^{\top}\right|} \in \mathbb{F}_{q^s}$  for any $V\in \mathcal{V}_q(k,n)$. Thus, $\eta\neq -\frac{\left|VM_{k}^{(h,k+t)}(\boldsymbol{\alpha})^{\top}\right|}{\left|VM_{k}(\boldsymbol{\alpha})^{\top}\right|}$ for  $\eta\notin \mathbb{F}_{q^s}$.
By Theorem \ref{mm1}, $\mathcal{C}_{1}$ is MRD.
\end{proof}

To determine the condition under which the twisted Gabidulin code $\mathcal{C}_{1}$ is not MRD, three useful lemmas are given below.
\begin{lemma}\label{lem1}
Let $$A_{t}=\left[\begin{array}{cccc}c_{0} & 0 & \cdots & 0 \\ c_{1}^{[1]} & c_{0}^{[1]} & \cdots & 0 \\ \vdots & \vdots & \ddots & \vdots \\ c_{t}^{[t]} & c_{t-1}^{[t]} & \cdots & c_{0}^{[t]}\end{array}\right],$$ where $c_{0}=1$ and $c_{1}, c_{2}, \ldots, c_{t} \in \mathbb{F}_{q^{m}}$. Then $$A_{t}^{-1}=\left[\begin{array}{cccc}
e_{0,0} & 0 & \cdots & 0 \\
e_{1,0} & e_{1,1} & \cdots & 0 \\
\vdots & \vdots & \ddots & \vdots \\
e_{t,0} & e_{t,1} & \cdots & e_{t,t}
\end{array}\right],$$  where $e_{0,0}=\ldots=e_{t,t}=1$ and $e_{i,j}=-\sum_{s=0}^{i-j-1} e_{j+s,j} c_{i-j-s}^{[i]}$ for $0\leq j<i\leq t$.
\end{lemma}
\begin{proof}
Denote $$B_{t}:=\left[\begin{array}{cccc}
e_{0,0} & 0 & \cdots & 0 \\
e_{1,0} & e_{1,1} & \cdots & 0 \\
\vdots & \vdots & \ddots & \vdots \\
e_{t,0} & e_{t,1} & \cdots & e_{t,t}
\end{array}\right],$$  where $e_{0,0}=\ldots=e_{t,t}=1$ and $e_{i,j}=-\sum_{s=0}^{i-j-1} e_{j+s,j} c_{i-j-s}^{[i]}$ for $0\leq j<i\leq t$. Let $V_{t}=A_{t}B_{t}$. Then  the entries of $V_{t}=[v_{i,j}]\in \mathbb{F}_{q^m}^{(t+1)\times (t+1)}$ are of the form
$$\begin{array}{ll}
v_{i, j}=(c_{i-1}^{[i-1]}\ \ldots \  c_{0}^{[i-1]}\   0\  \ldots \ 0)
 \cdot\left(\begin{array}{c}
0 \\
\vdots \\
0 \\
e_{j-1,j-1} \\
\vdots \\
e_{t,j-1}
\end{array}\right)= \left\{\begin{array}{ll}
0,& \text{if } i<j, \\
c_{0}^{[i-1]} e_{j-1,j-1}=1,& \text{if } i=j, \\
\sum_{s=0}^{i-j} c_{i-j-s}^{[i-1]} e_{j-1+s,j-1},& \text{if } i>j.
\end{array} \right. \\
\end{array}$$
Since $e_{i,j}=-\sum_{s=0}^{i-j-1} e_{j+s,j} c_{i-j-s}^{[i]}$ for $i> j$, we have $\sum_{s=0}^{i-j} c_{i-j-s}^{[i-1]} e_{j-1+s,j-1}=0$.
Therefore, $V_{t}$ is an identity matrix. Hence, $A_{t}^{-1}=B_{t}$.
\end{proof}

\begin{lemma}\label{lem3}
Suppose that $\alpha_{1},\ldots,\alpha_{k}\in \mathbb{F}_{q^m}$ are linearly independent over $\mathbb{F}_{q}$.  Let $\prod_{\alpha \in \left\langle\alpha_{1},\ldots,\alpha_{k}  \right\rangle_{\mathbb{F}_{q}}}(x-\alpha)=\sum_{j=0}^{k} c_{j} x^{[k-j]}$ and $c_j=0$ for $j>k$. Let $e_{t,0},\ldots, e_{t,t}$ be as defined in Lemma \ref{lem1} in terms of $c_0,\ldots,c_{t}$. Then
$$
\left|M_{k}^{(h,k+t)}(\boldsymbol{\alpha})\right|=g_{h}^{(t)}\left|M_{k}(\boldsymbol{\alpha})\right|,
$$
where $g_{h}^{(t)}=-\sum_{i=0}^{\min \{t, h\}} e_{t,i} c_{k-h+i}^{[i]}$ and $\boldsymbol{\alpha}=[\alpha_1,\ldots,\alpha_k]$.
\end{lemma}
\begin{proof}
Let $[g_{0}^{(t)},g_{1}^{(t)},\ldots, g_{k-1}^{(t)}]$ be the unique solution of the
following system of equations with unknowns $x_0,x_1,\ldots,x_{k-1}$:
$$[x_{0},x_{1},\ldots, x_{k-1}]\left[\begin{array}{cccc}
\alpha_{1} & \alpha_{2} & \ldots & \alpha_{k} \\
\alpha_{1}^{[1]} & \alpha_{2}^{[1]} & \ldots & \alpha_{k}^{[1]} \\
\vdots & \vdots & \ddots & \vdots \\
\alpha_{1}^{[k-1]} & \alpha_{2}^{[k-1]} & \ldots & \alpha_{k}^{[k-1]}
\end{array}\right]=[\alpha_{1}^{[k+t]},\alpha_{2}^{[k+t]},\ldots,\alpha_{k}^{[k+t]}].$$
Then we obtain $\alpha_{j}^{[k+t]}=\sum_{i=0}^{k-1} g_{i}^{(t)} \alpha_{j}^{[i]}$ for all $1 \leq j \leq k$. Let $f_{t}(x)=x^{[k+t]}-\sum_{i=0}^{k-1} g_{i}^{(t)} x^{[i]}$. Obviously, $\alpha_{1}, \ldots, \alpha_{k}$ are roots of  $f_{t}(x)$. %Since $\left\langle\alpha_{1},\ldots,\alpha_{k}  \right\rangle_{\mathbb{F}_{q}}$ is a $k$-dimensional linear space spanned by $\alpha_{1},\ldots,\alpha_{k}$ over $\mathbb{F}_{q}$,
Furthermore, each element of $\left\langle\alpha_{1},\ldots,\alpha_{k}  \right\rangle_{\mathbb{F}_{q}}$ is a root of $f_{t}(x)$.  By Lemma \ref{sp2}, there is $h_t(x)=\sum_{i=0}^{t} a_{i}^{(t)} x^{[i]}$ such that $f_{t}(x)=h_t(x) \circ \prod_{\alpha \in \left\langle\alpha_{1},\ldots,\alpha_{k}  \right\rangle_{\mathbb{F}_{q}}}(x-\alpha)$. %By Lemma \ref{sp1}, we can write $\prod_{\alpha \in U}(x-\alpha)=\sum_{j=0}^{k} c_{j} x^{[k-j]}$,
Thus,  we have
$$\begin{array}{ll}x^{[k+t]}-\sum_{i=0}^{k-1} g_{i}^{(t)} x^{[i]}&=\left(\sum_{i=0}^{t} a_{i}^{(t)} x^{[i]}\right)\circ\left(\sum_{j=0}^{k} c_{j} x^{[k-j]}\right)\\
&=\sum_{s=0}^{t+k}\left( \sum_{i=0}^{t}a_{i}^{(t)}c_{k-s+i}^{[i]}\right) x^{[s]},\end{array}$$ where $c_{k-s+i}=0$ when $s>k+i$.

Comparing the coefficients of the left side and the right side of the above equation, we obtain
$$\left\{\begin{array}{ll}
g_{r}^{(t)}=-\sum_{i=0}^{\min \{t, r\}} a_{i}^{(t)} c_{k-r+i}^{[i]},\ \ \  0\leq r\leq k-1, \\
\left[0,0, \ldots, 0,1\right]=\left[a_{0}^{(t)}, a_{1}^{(t)}, \ldots,a_{t}^{(t)}\right]\left[\begin{array}{cccc}c_{0} & 0 & \cdots & 0 \\ c_{1}^{[1]} & c_{0}^{[1]} & \cdots & 0 \\ \vdots & \vdots & \ddots & \vdots \\ c_{t}^{[t]} & c_{t-1}^{[t]} & \cdots & c_{0}^{[t]}\end{array}\right].
\end{array}\right.$$
By Lemma \ref{lem1}, we get
$$\begin{aligned}
\left[a_{0}^{(t)}, a_{1}^{(t)}, \ldots, a_{t}^{(t)}\right]& =[0,0, \ldots, 0,1]\left[\begin{array}{cccc}c_{0} & 0 & \cdots & 0 \\ c_{1}^{[1]} & c_{0}^{[1]} & \cdots & 0 \\ \vdots & \vdots & \ddots & \vdots \\ c_{t}^{[t]} & c_{t-1}^{[t]} & \cdots & c_{0}^{[t]}\end{array}\right]^{-1}\\
& =[0,0, \ldots, 0,1]\left[\begin{array}{cccc}
e_{0,0} & 0 & \cdots & 0 \\
e_{1,0} & e_{1,1} & \cdots & 0 \\
\vdots & \vdots & \ddots & \vdots \\
e_{t,0} & e_{t,1} & \cdots & e_{t,t}
\end{array}\right] \\
&=\left[e_{t,0}, e_{t,1}, \ldots, e_{t,t}\right].
\end{aligned}$$
%where $e_{t,0},\ldots, e_{t,t}$ are given in Lemma \ref{lem1}. %$e_{t,t}=1$ and $e_{t,j}=-\sum_{s=0}^{t-j-1}e_{j+s,j}c_{t-j-s}^{[t]}$ for all $0\leq j\leq t-1$.
Then
$$\begin{aligned}g_{r}^{(t)}&=-\sum_{i=0}^{\min \{t, r\}} a_{i}^{(t)} c_{k-r+i}^{[i]}=-\sum_{i=0}^{\min \{t, r\}} e_{t,i} c_{k-r+i}^{[i]},\end{aligned}$$
where $ 0 \leq r\leq k-1$.

For the $(h+1)$-th row of $M_{k}^{(h,k+t)}(\boldsymbol{\alpha})$, we replace $\alpha_{j}^{[k+t]}$ with $\sum_{i=0}^{k-1} g_{i}^{(t)} \alpha_{j}^{[i]}$  for all $1\leq j\leq k$. Since the terms containing $\alpha_{j}^{[i]}$  $(j\in \{1,\ldots,k\}, i\in\{0,\ldots,h-1\}\cup\{h+1,\ldots,k-1\})$ in the  $(h+1)$-th row can be eliminated by subtracting the other  $k-1$ rows, %combining Lemma \ref{dk},
we have
$$\begin{aligned}
\left|M_{k}^{(h,k+t)}(\boldsymbol{\alpha})\right|=\left|\begin{array}{cccc}
\alpha_{1} & \cdots & \alpha_{k} \\
\vdots &  & \vdots \\
\alpha_{1}^{[h-1]}  & \cdots & \alpha_{k}^{[h-1]} \\
\alpha_{1}^{[k+t]}  & \cdots & \alpha_{k}^{[k+t]} \\
\alpha_{1}^{[h+1]}  & \cdots & \alpha_{k}^{[h+1]} \\
\vdots  & & \vdots \\
\alpha_{1}^{[k-1]}  & \cdots & \alpha_{k}^{[k-1]}
\end{array}\right|
=g_{h}^{(t)}\left|\begin{array}{cccc}
\alpha_{1}  & \cdots & \alpha_{k} \\
\vdots &  & \vdots \\
\alpha_{1}^{[h-1]}  & \cdots & \alpha_{k}^{[h-1]} \\
\alpha_{1}^{[h]} & \cdots & \alpha_{k}^{[h]} \\
\alpha_{1}^{[h+1]}  & \cdots & \alpha_{k}^{[h+1]} \\
\vdots  & & \vdots \\
\alpha_{1}^{[k-1]}  & \cdots & \alpha_{k}^{[k-1]}
\end{array}\right|
=g_{h}^{(t)}\left|M_{k}(\boldsymbol{\alpha})\right|.\end{aligned}$$
It completes the proof.\end{proof}

\begin{remark}
From Lemma \ref{dk}, we have
$$
\left|M_{k}^{(h,k+t)}(\boldsymbol{\alpha})\right|=g_{h}^{(t)}\alpha_{1} \prod_{j=1}^{k-1} \prod_{b_{1}, \ldots, b_{j} \in \mathbb{F}_{q}}\left(\alpha_{j+1}-\sum_{i=1}^{j} b_{i} \alpha_{i}\right).
$$
In particular, $\left|M_{k}^{(h,k+t)}(\boldsymbol{\alpha})\right|\neq 0$ if and only if $g_{h}^{(t)}\neq 0$.
\end{remark}

\begin{lemma}\label{lemma1}
Suppose that $\alpha_{1},\ldots,\alpha_{k}\in \mathbb{F}_{q^m}$ are linearly independent over $\mathbb{F}_{q}$. Let $\prod_{\alpha \in \left\langle\alpha_{1},\ldots,\alpha_{k}  \right\rangle_{\mathbb{F}_{q}}}(x-\alpha)=\sum_{j=0}^{k} c_{j} x^{[k-j]}$, $c_j=0$ for $j>k$, and $e_{t,0},\ldots,e_{t,t}$  be as defined in Lemma \ref{lem1} in terms of $c_0,\ldots,c_{t}$. Let $G_{1k}$ be the submatrix formed by the first $k$ columns of $G_{1}$ given in (\ref{eqgt}), then $|G_{1k}|$=0 if and only if $\eta^{-1}=-g_{h}^{(t)}$, where $g_{h}^{(t)}=-\sum_{i=0}^{\min \{t, h\}} e_{t,i} c_{k-h+i}^{[i]}$.
\end{lemma}
\begin{proof}
It is clear that $G_{1k}$ is given by
$$G_{1k}=\left[\begin{array}{cccc}
\alpha_{1} & \alpha_{2} & \cdots & \alpha_{k} \\
\alpha_{1}^{[1]} & \alpha_{2}^{[1]} & \cdots & \alpha_{k}^{[1]} \\
\vdots &\vdots &  & \vdots \\
\alpha_{1}^{[h-1]} & \alpha_{2}^{[h-1]} & \cdots & \alpha_{k}^{[h-1]} \\
\alpha_{1}^{[h]}+\eta \alpha_{1}^{[k+t]} &\alpha_{2}^{[h]}+\eta \alpha_{2}^{[k+t]} & \cdots & \alpha_{k}^{[h]}+\eta \alpha_{k}^{[k+t]} \\
\alpha_{1}^{[h+1]} & \alpha_{2}^{[h+1]} & \cdots & \alpha_{k}^{[h+1]} \\
\vdots & \vdots & & \vdots \\
\alpha_{1}^{[k-1]} & \alpha_{2}^{[k-1]} & \cdots & \alpha_{k}^{[k-1]}
\end{array}\right].$$
Then by Lemma \ref{lem3}, we find that the determinant of $G_{1k}$ is
$$\begin{aligned}
|G_{1k}|=&\left|\begin{array}{cccc}
\alpha_{1} & \alpha_{2} & \cdots & \alpha_{k} \\
\vdots &\vdots &  & \vdots \\
\alpha_{1}^{[h-1]} & \alpha_{2}^{[h-1]} & \cdots & \alpha_{k}^{[h-1]} \\
\alpha_{1}^{[h]} &\alpha_{2}^{[h]} & \cdots & \alpha_{k}^{[h]} \\
\alpha_{1}^{[h+1]} & \alpha_{2}^{[h+1]} & \cdots & \alpha_{k}^{[h+1]} \\
\vdots & \vdots & & \vdots \\
\alpha_{1}^{[k-1]} & \alpha_{2}^{[k-1]} & \cdots & \alpha_{k}^{[k-1]}
\end{array}\right|+\eta \left|\begin{array}{cccc}
\alpha_{1} & \alpha_{2} & \cdots & \alpha_{k} \\
\vdots &\vdots &  & \vdots \\
\alpha_{1}^{[h-1]} & \alpha_{2}^{[h-1]} & \cdots & \alpha_{k}^{[h-1]} \\
\alpha_{1}^{[k+t]} & \alpha_{2}^{[k+t]} & \cdots & \alpha_{k}^{[k+t]} \\
\alpha_{1}^{[h+1]} & \alpha_{2}^{[h+1]} & \cdots & \alpha_{k}^{[h+1]} \\
\vdots & \vdots & & \vdots \\
\alpha_{1}^{[k-1]} & \alpha_{2}^{[k-1]} & \cdots & \alpha_{k}^{[k-1]}
\end{array}\right|\\
=&\left|M_{k}(\boldsymbol{\alpha})\right|+\eta\left|M_{k}^{(h,k+t)}(\boldsymbol{\alpha})\right|\\
=&(1+\eta g_{h}^{(t)})\cdot \left|M_{k}(\boldsymbol{\alpha})\right|,\end{aligned}$$%\label{det}
where $\boldsymbol{\alpha}=[\alpha_1,\ldots,\alpha_k]$. By Lemma \ref{dk}, since $\alpha_{1},\ldots,\alpha_{k}\in \mathbb{F}_{q^m}$ are linearly independent over $\mathbb{F}_{q}$, we have $\left|M_{k}(\boldsymbol{\alpha})\right|\neq 0$.
Then $|G_{1k}|$=0 if and only if $1+\eta g_{h}^{(t)}=0$, i.e., $\eta^{-1}=-g_{h}^{(t)}$. It completes the proof.
\end{proof}
From the above results, we can obtain the condition that $\mathcal{C}_{1}$ is not an MRD code.
\begin{theorem}\label{tno2}
Let $g_{h}^{(t)}=-\sum_{i=0}^{\min \{t, h\}} e_{t,i} c_{k-h+i}^{[i]}$, where $e_{t,0},\ldots,e_{t,t}$  are as defined in Lemma \ref{lem1} in terms of $c_0,\ldots,c_{t}$, and $c_{0},\ldots,c_{k}$ satisfy $\prod_{\alpha \in \left\langle\alpha_{1},\ldots,\alpha_{k}  \right\rangle_{\mathbb{F}_{q}}}(x-\alpha)=\sum_{j=0}^{k} c_{j} x^{[k-j]}$, and $c_j=0$ for $j>k$. If $\eta^{-1}=-g_{h}^{(t)}$, then $\mathcal{C}_{1}$  is not MRD.
\end{theorem}
\begin{proof}
By Lemma \ref{nmrd}, we know that if a generator matrix of the code $\mathcal{C}_{1}$ has at least one zero maximal minor, then $\mathcal{C}_{1}$ is not an MRD code.
Let $G_{1}$, given in (\ref{eqgt}), be a generator matrix of $\mathcal{C}_{1}$.
By Lemma \ref{lemma1}, if $\eta^{-1}=-g_{h}^{(t)}$, then the $k\times k$ maximal minor formed by the first $k$ columns of $G_{1}$ is equal to $0$, and so  $\mathcal{C}_{1}$  is not MRD.
\end{proof}

\begin{definition}
Suppose  that $\alpha_{1},\ldots,\alpha_{n}\in \mathbb{F}_{q^m}$ are linearly independent over $\mathbb{F}_{q}$. For any k-subset $\mathcal{I}\subseteq\textit{\textbf{[}}n\textit{\textbf{]}}$, let $U_{\mathcal{I}}=\left\langle\alpha_{i} \ (i\in \mathcal{I}) \right\rangle_{\mathbb{F}_{q}}$, $\prod_{\alpha \in U_{\mathcal{I}}}(x-\alpha)=\sum_{j=0}^{k} c_{j} x^{[k-j]}$, and $c_j=0$ for $j>k$. Let  $e_{t,0},\ldots,e_{t,t}$ be as defined in Lemma \ref{lem1} in terms of $c_0,\ldots,c_{t}$ and $g_{h}^{(t)}(\mathcal{I})=-\sum_{i=0}^{\min \{t, h\}} e_{t,i} c_{k-h+i}^{[i]}$. We denote
$$\Omega_1:=\left\{-g_{h}^{(t)}(\mathcal{I}):   \mathcal{I} \text{ is a } k\text{-subset of } \textit{\textbf{[}}n\textit{\textbf{]}} \right\}.$$
\end{definition}
Using the same argument in the proofs of Lemma \ref{lemma1} and Theorem \ref{tno2}, we obtain the following result.
\begin{theorem}
If $\eta^{-1}\in\Omega_1$, then $\mathcal{C}_{1}$  is not MRD.
\end{theorem}
\begin{proof}
Let $G_{1}$, given in (\ref{eqgt}), be a generator matrix of $\mathcal{C}_{1}$. If $\eta^{-1}\in\Omega_1$, then there exists a $k$-subset $\mathcal{I}\subseteq\textit{\textbf{[}}n\textit{\textbf{]}}$ such that the determinant of the $k\times k$ submatrix formed by the columns of $G_{1}$ indexed by $\mathcal{I}$ is zero. Hence, $\mathcal{C}_{1}$  is not MRD.  
\end{proof}

Below we consider a special case. Suppose that $t=0$ in (\ref{p1}), we obtain
\begin{equation}\label{stp1}
\mathcal{P}_{1}=\left\{f(x)=\sum_{i=0}^{k-1} f_i x^{[i]}+\eta f_{h} x^{[k]}: f_i \in \mathbb{F}_{q^m},0\leq i\leq k-1\right\},
\end{equation}
where $0\leq h\leq k-1$. In this case, we have
$$g_{h}^{(t)}=g_{h}^{(0)}=-e_{0,0} c_{k-h}=-c_{k-h}.$$
Then we obtain the following corollary.
\begin{corollary}\label{t0}
Let $\mathcal{P}_{1}$ be given in (\ref{stp1}), and let $\boldsymbol{\alpha}=[\alpha_{1},\ldots,\alpha_{n}]\in \mathbb{F}_{q^{m}}^{n}$ with $\mathrm{rk}_{\mathbb{F}_{q}}(\boldsymbol{\alpha})=n$. For any $k$-subset $\mathcal{I}\subseteq\textit{\textbf{[}}n\textit{\textbf{]}}$, let $U_{\mathcal{I}}=\left\langle\alpha_{i} \ (i\in \mathcal{I}) \right\rangle_{\mathbb{F}_{q}}$. If $\eta^{-1}\in \Omega_{1}^{\prime}$, where
$$\Omega_{1}^{\prime}=\left\{c_{k-h}: \mathcal{I} \text{ is a } k\text{-subset of } \textit{\textbf{[}}n\textit{\textbf{]}}, \prod_{\alpha \in U_{\mathcal{I}}}(x-\alpha)=\sum_{j=0}^{k} c_{j} x^{[k-j]} \right\},$$
then $\mathcal{C}_{1}=\mathrm{ev}_{\boldsymbol{\alpha}}(\mathcal{P}_{1})$ is not MRD.
\end{corollary}

For $t=0$ and $0\leq h\leq k-1$, from \cite[Proposition 4]{fiz} and \cite[Lemma 3]{shee}, we obtain the following propositions.
\begin{proposition}\label{hd1}
Let $\mathcal{P}_{1}$ be given in (\ref{stp1}), and let $\boldsymbol{\alpha}\in \mathbb{F}_{q^{m}}^{n}$ with $\mathrm{rk}_{\mathbb{F}_{q}}(\boldsymbol{\alpha})=n$.  For all $f(x)\in\mathcal{P}_{1}$ of $q$-degree $k$, if $f_{0} z\neq(-1)^{k} \eta^{q}f_{h}^{q} z^{q}$ for any $z\in \mathbb{F}_{q^m}^{*}$, then $\mathcal{C}_1=\mathrm{ev}_{\boldsymbol{\alpha}}\left(\mathcal{P}_{1}\right)$ is MRD.
\end{proposition}
\begin{proof}
It follows from \cite[Proposition 4]{fiz} that if the kernel of a linearized polynomial $f(x)=\sum_{i=0}^{k} f_i x^{[i]}$ of $q$-degree $k$ over $\mathbb{F}_{q^m}$ has dimension $k$, then there is $z\in \mathbb{F}_{q^m}^{*}$ such that $f_{0} z=(-1)^{k} f_{k}^{q} z^{q}$. By setting $f_k=\eta f_h$, we obtain a twisted linearized polynomial $f(x)=\sum_{i=0}^{k-1} f_i x^{[i]}+\eta f_h x^{[k]}$ which belongs to $\mathcal{P}_{1}$. In this case, the above equation becomes $f_{0} z=(-1)^{k} \eta^{q}f_{h}^{q} z^{q}$.

Thus, for $f(x)\in\mathcal{P}_{1}$ of $q$-degree $k$, if $f_{0} z\neq(-1)^{k} \eta^{q}f_{h}^{q} z^{q}$ for any $z\in \mathbb{F}_{q^m}^{*}$, then the dimension of the kernel of $f(x)$ is at most $k-1$, and so $\mathcal{C}_1=\mathrm{ev}_{\boldsymbol{\alpha}}\left(\mathcal{P}_{1}\right)$ is MRD.
\end{proof}
\begin{proposition}\label{hd2}
Let $\mathcal{P}_{1}$ be given in (\ref{stp1}), and let $\boldsymbol{\alpha}\in \mathbb{F}_{q^{m}}^{n}$ with $\mathrm{rk}_{\mathbb{F}_{q}}(\boldsymbol{\alpha})=n$.  For all $f(x)\in\mathcal{P}_{1}$ of $q$-degree $k$, if $N(f_0)\neq(-1)^{mk}N(\eta)N(f_h)$, where $N(\cdot)$ denote the field norm from $\mathbb{F}_{q^m}$ to $\mathbb{F}_q$ with $N(x)=x^{1+q+\ldots+q^{m-1}}$, then $\mathcal{C}_1=\mathrm{ev}_{\boldsymbol{\alpha}}\left(\mathcal{P}_{1}\right)$ is MRD.
\end{proposition}
\begin{proof}
It follows from \cite[Lemma 3]{shee} that if a linearized polynomial $f(x)=\sum_{i=0}^{k} f_i x^{[i]}$ of $q$-degree $k$ over $\mathbb{F}_{q^m}$ has rank $n-k$, i.e., the rank weight of the corresponding codeword is $n-k$, then $N(f_0)=(-1)^{mk}N(f_k)$. By setting $f_k=\eta f_h$, we get a twisted linearized polynomial $f(x)=\sum_{i=0}^{k-1} f_i x^{[i]}+\eta f_h x^{[k]}$ belonging to $\mathcal{P}_{1}$, and the above equation becomes $N(f_0)=(-1)^{mk}N(\eta)N(f_h)$.

We know that the kernel of $f(x)$ of $q$-degree $k$ has dimension at most $k$, implying that the rank weight of the corresponding codeword is at least $n-k$. Hence, for each $f(x)\in\mathcal{P}_{1}$ of $q$-degree $k$, if $N(f_0)\neq(-1)^{mk}N(\eta)N(f_h)$, then the rank weight of the corresponding codeword is at least $n-k+1$, and so $\mathcal{C}_1=\mathrm{ev}_{\boldsymbol{\alpha}}\left(\mathcal{P}_{1}\right)$ is MRD.
\end{proof}

\begin{remark}
 When $h=0$, the condition in Proposition \ref{hd2} becomes $N(\eta)\neq(-1)^{mk}$ as shown in \cite{shee,GTG,otal}.
\end{remark}
\begin{remark}
In fact, the condition of Proposition \ref{hd2} can imply that of Proposition \ref{hd1}, i.e., if $N(f_0)\neq(-1)^{mk}N(\eta)N(f_h)$, then $f_{0} z\neq(-1)^{k} \eta^{q}f_{h}^{q} z^{q}$ for any $z\in \mathbb{F}_{q^m}^{*}$. Note that the converse is not true.
\end{remark}

In the following, we study twisted Gabidulin codes in Hamming metric. Using the same argument in the proof of Lemma \ref{lem3}, we establish the necessary and sufficient condition for $\mathcal{C}_{1}$ to be MDS.

\begin{theorem}
 $\mathcal{C}_{1}$ is MDS if and only if $\eta^{-1}\in \mathbb{F}_{q^m}^{*}\backslash \Omega_1$.
\end{theorem}
\begin{proof}
Let $G_{1}$, as defined in (\ref{eqgt}), be a generator matrix of $\mathcal{C}_{1}$.
 By Proposition \ref{gmds},
$\mathcal{C}_{1}$  is  MDS  if and only if for each  $k$-subset $\mathcal{I}\subseteq\textit{\textbf{[}}n\textit{\textbf{]}}$, the $k\times k$ minor formed by the columns of $G_{1}$ indexed by $\mathcal{I}$ is non-zero, if and only if $\eta^{-1} \notin \Omega_1$. It completes the proof.
\end{proof}

We now recall the equivalent conditions of NMDS codes and AMDS codes.

\begin{lemma}[\cite{3}]\label{anmds}
An  $[n, k]$  linear code  $\mathcal{C}$  over  $\mathbb{F}_{q^m}$  is NMDS  if and only if the generator matrix  $G$  of  $\mathcal{C}$  satisfies the following conditions:
\begin{enumerate}[(i)]
\item Any $k-1$ columns of $G$ are linearly independent.
\item There exists  $k$  linearly dependent columns in  $G$.
\item Any  $k+1$  columns of $G$ are of rank $k$.
\end{enumerate}
The code $\mathcal{C}$ is AMDS if and only if the above two conditions (ii) and (iii) hold.
\end{lemma}

According to Lemma \ref{anmds}, we present the sufficient and necessary conditions for $\mathcal{C}_1$ to be AMDS or NMDS.

\begin{theorem}\label{amds1}
Let $\eta^{-1}\in \Omega_1$. Then $\mathcal{C}_1$  is AMDS if and only if for each $(k+1)$-subset  $\mathcal{J} \subseteq\textit{\textbf{[}}n\textit{\textbf{]}}$, there exists a $k$-subset  $\mathcal{I} \subseteq \mathcal{J}$ such that $\eta^{-1}\neq -g_{h}^{(t)}(\mathcal{I})$.
\end{theorem}
\begin{proof}
Since $\eta^{-1}\in \Omega_1$, the condition (ii) of Lemma \ref{anmds} holds. Thus, the code $\mathcal{C}_1$ is AMDS if and only if the condition (iii) of Lemma \ref{anmds} holds,  if and only if for each $(k+1)$-subset  $\mathcal{J} \subseteq\textit{\textbf{[}}n\textit{\textbf{]}}$, there exists a $k$-subset  $\mathcal{I} \subseteq \mathcal{J}$ such that $\eta^{-1}\neq -g_{h}^{(t)}(\mathcal{I})$.
\end{proof}
\begin{theorem}\label{nmds1}
Let $\eta^{-1}\in \Omega_1$ and $h\in \{0,k-1\}$. Then  $\mathcal{C}_1$ is NMDS if and only if for each $(k+1)$-subset  $\mathcal{J} \subseteq\textit{\textbf{[}}n\textit{\textbf{]}}$, there exists a $k$-subset  $\mathcal{I} \subseteq \mathcal{J}$ such that $\eta^{-1}\neq -g_{h}^{(t)}(\mathcal{I})$.
\end{theorem}
\begin{proof}
Since $h\in \{0,k-1\}$, the condition (i) of Lemma \ref{anmds} holds.
Since $\eta^{-1}\in \Omega_1$, the condition (ii) of Lemma \ref{anmds} holds. Therefore, $\mathcal{C}_1$ is NMDS if and only if the condition (iii) of Lemma \ref{anmds} holds,  if and only if for each $(k+1)$-subset  $\mathcal{J} \subseteq\textit{\textbf{[}}n\textit{\textbf{]}}$, there exists a $k$-subset  $\mathcal{I} \subseteq \mathcal{J}$ such that $\eta^{-1}\neq -g_{h}^{(t)}(\mathcal{I})$.
\end{proof}

\section{Twisted Gabidulin codes with two twists}
By adding two monomials to the same position of each linearized polynomial $f(x)=\sum_{i=0}^{k-1} f_i x^{[i]}$ in the definition of Gabidulin codes, we obtain twisted Gabidulin codes with two twists which are defined as follows.
\begin{definition}\label{l2}
Let $n, k$ be positive integers with $k<n\leq m$ and $\boldsymbol{\eta}=[\eta_{1},\eta_{2}]\in (\mathbb{F}_{q^m}^{\ast})^{2}$. Denote $\mathcal{P}_{2}$ as the set of twisted linearized polynomials over $\mathbb{F}_{q^{m}}$ given by
\begin{equation}\label{p2}
\mathcal{P}_{2}:=\left\{f(x)=\sum_{i=0}^{k-1} f_i x^{[i]}+\eta_{1} f_{h}x^{[k+t_{1}]}+\eta_{2}f_{h} x^{[k+t_2]}: f_i \in \mathbb{F}_{q^m},0\leq i\leq k-1\right\},
\end{equation}
where $0\leq h\leq k-1$ and $0\leq t_{1}<t_{2}\leq n-k-1$.  Let $\boldsymbol{\alpha}=[\alpha_{1},\ldots,\alpha_{n}]\in \mathbb{F}_{q^{m}}^{n}$ with $\mathrm{rk}_{\mathbb{F}_{q}}(\boldsymbol{\alpha})=n$. The twisted Gabidulin code with two twists is defined by
\begin{equation}\label{c2}
\mathcal{C}_{2}:=\mathrm{ev}_{\boldsymbol{\alpha}}\left(\mathcal{P}_{2}\right)\subseteq \mathbb{F}_{q^m}^n.
\end{equation}
\end{definition}
Unless otherwise stated, let $\mathcal{C}_{2}$ be the twisted Gabidulin code defined as in (\ref{c2}). Now we give two useful propositions about the code $\mathcal{C}_{2}$.

\begin{proposition}
$\mathcal{C}_{2}$ is an $[n, k]$ linear code over $\mathbb{F}_{q^{m}}$ with the generator matrix
\begin{equation}\label{eqgt2}
G_{2}=\left[\begin{array}{ccc}
\alpha_{1} & \cdots & \alpha_{n} \\
\vdots &  & \vdots \\
\alpha_{1}^{[h-1]} & \cdots & \alpha_{n}^{[h-1]} \\
\alpha_{1}^{[h]}+\eta_{1} \alpha_{1}^{[k+t_{1}]}+\eta_{2} \alpha_{1}^{[k+t_{2}]} & \cdots & \alpha_{n}^{[h]}+\eta_{1} \alpha_{n}^{[k+t_{1}]}+\eta_{2} \alpha_{n}^{[k+t_{2}]} \\
\alpha_{1}^{[h+1]} & \cdots & \alpha_{n}^{[h+1]} \\
\vdots  & & \vdots \\
\alpha_{1}^{[k-1]} & \cdots & \alpha_{n}^{[k-1]}
\end{array}\right].\end{equation}
\end{proposition}
\begin{proposition}
For the twisted Gabidulin code $\mathcal{C}_{2}$, we have $$\mathcal{C}_{2}=\left\langle\boldsymbol{\alpha}^{[h]}+\eta_{1} \boldsymbol{\alpha}^{[k+t_{1}]}+\eta_{2} \boldsymbol{\alpha}^{[k+t_{2}]}, \boldsymbol{\alpha}^{[s]}(s\in \{0,\ldots,k-1\}\backslash\{h\})  \right\rangle_{\mathbb{F}_{q^m}}.$$
\end{proposition}

Next we establish the necessary and sufficient condition for the twisted Gabidulin code $\mathcal{C}_{2}$ to be MRD.

\begin{theorem}\label{mm2} $\mathcal{C}_{2}$ is MRD if and only if $\boldsymbol{\eta}=[\eta_{1}, \eta_{2}] \in \mathcal{N}$, where
$$
\mathcal{N}=\left\{\boldsymbol{\eta} \in (\mathbb{F}_{q^{m}}^{*})^{2}: \forall V\in \mathcal{V}_q(k,n), \sum_{i=1}^{2}\eta_i\left|VM_{k}^{(h,k+t_i)}(\boldsymbol{\alpha})^{\top}\right| \neq -\left|V M_{k}(\boldsymbol{\alpha})^{\top}\right|\right\}.
$$
\end{theorem}
\begin{proof}
Let $G_2$ given in (\ref{eqgt2}) be a generator matrix of the code  $\mathcal{C}_2$. For  $V\in \mathcal{V}_{q}(k, n)$,  we obtain
$$\begin{aligned}
\left|V G_2^{\top}\right|&=\left|V\left[\begin{array}{cccc}
\boldsymbol{\alpha} \\
\vdots  \\
\boldsymbol{\alpha}^{[h-1]} \\
\boldsymbol{\alpha}^{[h]} \\
\boldsymbol{\alpha}^{[h+1]}  \\
\vdots \\
\boldsymbol{\alpha}^{[k-1]}
\end{array}\right]^{\top}\right|+\left|V\left[\begin{array}{cccc}
\boldsymbol{\alpha} \\
\vdots  \\
\boldsymbol{\alpha}^{[h-1]} \\
\eta_1 \boldsymbol{\alpha}^{[k+t_1]} \\
\boldsymbol{\alpha}^{[h+1]}  \\
\vdots \\
\boldsymbol{\alpha}^{[k-1]}
\end{array}\right]^{\top}\right|+\left|V\left[\begin{array}{cccc}
\boldsymbol{\alpha} \\
\vdots  \\
\boldsymbol{\alpha}^{[h-1]} \\
\eta_2 \boldsymbol{\alpha}^{[k+t_2]} \\
\boldsymbol{\alpha}^{[h+1]}  \\
\vdots \\
\boldsymbol{\alpha}^{[k-1]}
\end{array}\right]^{\top}\right|\\
&=\left|V M_{k}(\boldsymbol{\alpha})^{\top}\right|+\eta_1\left|VM_{k}^{(h,k+t_1)}(\boldsymbol{\alpha})^{\top}\right|+\eta_2\left|VM_{k}^{(h,k+t_2)}(\boldsymbol{\alpha})^{\top}\right|.
\end{aligned}$$
Thus, by Proposition \ref{gmrd}, $\mathcal{C}_{2}$ is MRD if and only if $\left|V G_2^{\top}\right|\neq 0$, if and only if
$$\eta_1\left|VM_{k}^{(h,k+t_1)}(\boldsymbol{\alpha})^{\top}\right|+\eta_2\left|VM_{k}^{(h,k+t_2)}(\boldsymbol{\alpha})^{\top}\right| \neq -\left|V M_{k}(\boldsymbol{\alpha})^{\top}\right|,$$
that is,
$[\eta_{1}, \eta_{2}] \in \mathcal{N}$.
\end{proof}

Below we construct a class of MRD twisted Gabidulin codes with two twists.
\begin{proposition}
Let $s$ be a positive integer with $n\leq s$ such that $\mathbb{F}_{q} \subsetneq \mathbb{F}_{q^{s}} \subsetneq \mathbb{F}_{q^m}$ is a chain of subfields. Let $\alpha_{1}, \ldots, \alpha_{n}\in \mathbb{F}_{q^s}$ be linearly independent over $\mathbb{F}_{q}$ and  $\boldsymbol{\eta}=[\eta_1,\eta_2]$ with $\eta_2=b\eta_1$, where $\eta_1\in \mathbb{F}_{q^m} \backslash \mathbb{F}_{q^s}$ and $b\in \mathbb{F}_{q^s}^*$. Then  $\mathcal{C}_{2}$  is MRD.
\end{proposition}
\begin{proof}
For any $V\in \mathcal{V}_{q}(k, n)$, we have $\left|V M_{k}(\boldsymbol{\alpha})^{\top}\right|\in \mathbb{F}_{q^{s}}^*$ since $\alpha_{1}, \ldots, \alpha_{n}\in \mathbb{F}_{q^{s}}$. For $\eta_1\in \mathbb{F}_{q^m} \backslash \mathbb{F}_{q^s}$ and $b\in \mathbb{F}_{q^s}^*$,
$$
\sum_{i=1}^2\eta_i\left|VM_{k}^{(h,k+t_i)}(\boldsymbol{\alpha})^{\top}\right|
=\eta_1\left(\left|VM_{k}^{(h,k+t_1)}(\boldsymbol{\alpha})^{\top}\right|+b\left|VM_{k}^{(h,k+t_2)}(\boldsymbol{\alpha})^{\top}\right|\right)\notin \mathbb{F}_{q^{s}}^*
$$
since $\left|VM_{k}^{(h,k+t_1)}(\boldsymbol{\alpha})^{\top}\right|+b\left|VM_{k}^{(h,k+t_2)}(\boldsymbol{\alpha})^{\top}\right|\in \mathbb{F}_{q^{s}}.$
Therefore, we have $$\sum_{i=1}^2\eta_i\left|VM_{k}^{(h,k+t_i)}(\boldsymbol{\alpha})^{\top}\right| \neq -\left|V M_{k}(\boldsymbol{\alpha})^{\top}\right|.$$
By Theorem \ref{mm2}, $\mathcal{C}_{2}$ is MRD. 
\end{proof}

For the following construction of MRD codes which is a special case of \cite{furt}, we can provide a direct proof by using Theorem \ref{mm2}.
\begin{proposition}\label{777}
Let $s_1, s_2$ be  positive integers with $n\leq s_1$ such that $\mathbb{F}_{q} \subsetneq \mathbb{F}_{q^{s_1}}\subsetneq \mathbb{F}_{q^{s_2}} \subsetneq \mathbb{F}_{q^{s_3}}=\mathbb{F}_{q^m}$ is a chain of subfields. Let $\alpha_{1}, \ldots, \alpha_{n}\in \mathbb{F}_{q^{s_1}}$ be linearly independent over $\mathbb{F}_{q}$ and  $\eta_{i}\in \mathbb{F}_{q^{s_{i+1}}} \backslash \mathbb{F}_{q^{s_{i}}}$ for $i=1,2$. Then  $\mathcal{C}_{2}$  is MRD.
\end{proposition}
\begin{proof}
For $\alpha_{1}, \ldots, \alpha_{n}\in \mathbb{F}_{q^{s_1}}$ and $\eta_{1}\in \mathbb{F}_{q^{s_{2}}} \backslash \mathbb{F}_{q^{s_{1}}}$, we have $$\left|V M_{k}(\boldsymbol{\alpha})^{\top}\right|\in \mathbb{F}_{q^{s_{1}}}^*\ \ \text{and}\ \ \eta_1\left|VM_{k}^{(h,k+t_1)}(\boldsymbol{\alpha})^{\top}\right|\in \mathbb{F}_{q^{s_{2}}} \backslash \mathbb{F}_{q^{s_{1}}}^*.$$
Denote $A:=\eta_1\left|VM_{k}^{(h,k+t_1)}(\boldsymbol{\alpha})^{\top}\right|+\eta_2\left|VM_{k}^{(h,k+t_2)}(\boldsymbol{\alpha})^{\top}\right|$.
If  $\left|VM_{k}^{(h,k+t_2)}(\boldsymbol{\alpha})^{\top}\right|=0$,
then $A=\eta_1\left|VM_{k}^{(h,k+t_1)}(\boldsymbol{\alpha})^{\top}\right|\notin \mathbb{F}_{q^{s_{1}}}^*$. Otherwise, we obtain
$$\eta_2= \frac{A-\eta_1\left|VM_{k}^{(h,k+t_1)}(\boldsymbol{\alpha})^{\top}\right|}{\left|VM_{k}^{(h,k+t_2)}(\boldsymbol{\alpha})^{\top}\right|}\notin \mathbb{F}_{q^{s_{2}}}.$$
Therefore, $A\notin \mathbb{F}_{q^{s_{2}}}$. In particular, $A\notin \mathbb{F}_{q^{s_{1}}}^*$. Hence, $A\neq -\left|V M_{k}(\boldsymbol{\alpha})^{\top}\right|$. By Theorem \ref{mm2}, $\mathcal{C}_{2}$ is MRD.\end{proof}

The following lemma plays an important role in determining when the code $\mathcal{C}_{2}$ is not MRD.
\begin{lemma}\label{lemma2}
Suppose that $\alpha_{1},\ldots,\alpha_{k}\in \mathbb{F}_{q^m}$ are linearly independent over $\mathbb{F}_{q}$. Let $\prod_{\alpha \in \left\langle\alpha_{1},\ldots,\alpha_{k}  \right\rangle_{\mathbb{F}_{q}}}(x-\alpha)=\sum_{j=0}^{k} c_{j} x^{[k-j]}$ and $c_j=0$ for $j>k$. Let $G_{2k}$ be the submatrix  formed by the first $k$ columns of $G_{2}$ given in (\ref{eqgt2}), then $\left|G_{2k}\right|$=0 if and only if %$\eta_{1}, \eta_{2}$ satisfy
$1+\eta_{1} g_{h}^{(t_{1})}+\eta_{2} g_{h}^{(t_{2})}=0$, where $g_{h}^{(t_{j})}=-\sum_{i=0}^{\min \{t_{j}, h\}} e_{t_{j},i} c_{k-h+i}^{[i]}$ and $e_{t_{j},0},\ldots,e_{t_{j},t_{j}}$ are as defined in Lemma \ref{lem1} for $j=1,2$.
\end{lemma}
\begin{proof}
By applying Lemma \ref{lem3}, we have
$$\begin{aligned}
 &\ |G_{2k}|\\
=&\left|\begin{array}{ccc}
\alpha_{1} & \cdots & \alpha_{k} \\
\vdots &  & \vdots \\
\alpha_{1}^{[h-1]} & \cdots & \alpha_{k}^{[h-1]} \\
\alpha_{1}^{[h]}+\eta_{1} \alpha_{1}^{[k+t_{1}]}+\eta_{2} \alpha_{1}^{[k+t_{2}]} & \cdots & \alpha_{k}^{[h]}+\eta_{1} \alpha_{k}^{[k+t_{1}]}+\eta_{2} \alpha_{k}^{[k+t_{2}]} \\
\alpha_{1}^{[h+1]}  & \cdots & \alpha_{k}^{[h+1]} \\
\vdots  & & \vdots \\
\alpha_{1}^{[k-1]}  & \cdots & \alpha_{k}^{[k-1]}
\end{array}\right|\\
=&\left|\begin{array}{ccc}
\alpha_{1}  & \cdots & \alpha_{k} \\
\vdots  &  & \vdots \\
\alpha_{1}^{[h-1]}  & \cdots & \alpha_{k}^{[h-1]} \\
\alpha_{1}^{[h]}  & \cdots & \alpha_{k}^{[h]} \\
\alpha_{1}^{[h+1]}  & \cdots & \alpha_{k}^{[h+1]} \\
\vdots  & &\vdots \\
\alpha_{1}^{[k-1]} & \cdots & \alpha_{k}^{[k-1]}
\end{array}\right|+\eta_{1} \left|\begin{array}{ccc}
\alpha_{1}  & \cdots & \alpha_{k} \\
\vdots  &  & \vdots \\
\alpha_{1}^{[h-1]}  & \cdots & \alpha_{k}^{[h-1]} \\
\alpha_{1}^{[k+t_{1}]}  & \cdots & \alpha_{k}^{[k+t_{1}]} \\
\alpha_{1}^{[h+1]}  & \cdots & \alpha_{k}^{[h+1]} \\
\vdots  & & \vdots \\
\alpha_{1}^{[k-1]} & \cdots & \alpha_{k}^{[k-1]}
\end{array}\right|+\eta_{2} \left|\begin{array}{ccc}
\alpha_{1}  & \cdots & \alpha_{k} \\
\vdots  &  & \vdots \\
\alpha_{1}^{[h-1]}  & \cdots & \alpha_{k}^{[h-1]} \\
\alpha_{1}^{[k+t_{2}]}  & \cdots & \alpha_{k}^{[k+t_{2}]} \\
\alpha_{1}^{[h+1]}  & \cdots & \alpha_{k}^{[h+1]} \\
\vdots  & & \vdots \\
\alpha_{1}^{[k-1]} & \cdots & \alpha_{k}^{[k-1]}
\end{array}\right|\\
=&(1+\eta_{1} g_{h}^{(t_{1})}+\eta_{2} g_{h}^{(t_{2})})\cdot \left|M_{k}(\boldsymbol{\alpha})\right|.\end{aligned}$$
Then by Lemma \ref{dk}, $|G_{2k}|$=0 if and only if $1+\eta_{1} g_{h}^{(t_{1})}+\eta_{2} g_{h}^{(t_{2})}=0$. It completes the proof.
\end{proof}
Now we derive the condition that the twisted Gabidulin code $\mathcal{C}_{2}$ is not an MRD code.
\begin{theorem}\label{tno4}
Let $g_{h}^{(t_{j})}=-\sum_{i=0}^{\min \{t_{j}, h\}} e_{t_{j},i} c_{k-h+i}^{[i]}$ for $j=1,2$, where $e_{t_{j},0},\ldots,e_{t_{j},t_{j}}$ are as defined in Lemma \ref{lem1}, $c_{0},\ldots,c_{k}$ satisfy $\prod_{\alpha \in \left\langle\alpha_{1},\ldots,\alpha_{k}  \right\rangle_{\mathbb{F}_{q}}}(x-\alpha)=\sum_{j=0}^{k} c_{j} x^{[k-j]}$, and $c_j=0$ for $j>k$. If $1+\eta_{1} g_{h}^{(t_{1})}+\eta_{2} g_{h}^{(t_{2})}=0$, then $\mathcal{C}_{2}$  is not MRD.
\end{theorem}
\begin{proof}
Let $G_{2}$, given in (\ref{eqgt2}), be a generator matrix of $\mathcal{C}_{2}$.
By Lemma \ref{lemma2}, if $1+\eta_{1} g_{h}^{(t_{1})}+\eta_{2} g_{h}^{(t_{2})}=0$, then the $k\times k$ minor formed by the first $k$ columns of $G_{2}$ is zero. By Lemma \ref{nmrd},  $\mathcal{C}_{2}$  is not MRD.
\end{proof}

\begin{definition}
Suppose  that $\alpha_{1},\ldots,\alpha_{n}\in \mathbb{F}_{q^m}$ are linearly independent over $\mathbb{F}_{q}$. For any $k$-subset $\mathcal{I}\subseteq \textit{\textbf{[}}n\textit{\textbf{]}}$, let $U_{\mathcal{I}}=\left\langle\alpha_{i} \ (i\in \mathcal{I}) \right\rangle_{\mathbb{F}_{q}}$, $\prod_{\alpha \in U_{\mathcal{I}}}(x-\alpha)=\sum_{j=0}^{k} c_{j} x^{[k-j]}$, and $c_j=0$ for $j>k$. Let $e_{t_{j},0},\ldots,e_{t_{j},t_{j}}$ be as defined in Lemma \ref{lem1} and $g_{h}^{(t_{j})}(\mathcal{I})=-\sum_{i=0}^{\min \{t_{j}, h\}} e_{t_{j},i} c_{k-h+i}^{[i]}$ for $j=1,2$.  We denote $\Omega_2:=$
$$
\left\{[\eta_{1}, \eta_{2}] \in (\mathbb{F}_{q^{m}}^{*})^{2}:\      1+\eta_{1} g_{h}^{(t_{1})}(\mathcal{I})+\eta_{2} g_{h}^{(t_{2})}(\mathcal{I})=0\ \text{for some}\ k\text{-subset}\ \mathcal{I}\text{ of } \textit{\textbf{[}}n\textit{\textbf{]}}\right\}.
$$
\end{definition}

Using the same argument in the proofs of Lemma \ref{lemma2} and Theorem \ref{tno4}, we obtain the following result.
\begin{theorem}
If $\boldsymbol{\eta}=[\eta_{1}, \eta_{2}] \in\Omega_2$, then $\mathcal{C}_{2}$  is not MRD.
\end{theorem}
\begin{proof}
Let $G_{2}$, given in (\ref{eqgt2}), be a generator matrix of $\mathcal{C}_{2}$.
If $\boldsymbol{\eta}\in\Omega_2$, then there exists a $k$-subset $\mathcal{I}\subseteq\textit{\textbf{[}}n\textit{\textbf{]}}$ such that $1+\eta_{1} g_{h}^{(t_{1})}(\mathcal{I})+\eta_{2} g_{h}^{(t_{2})}(\mathcal{I})=0$, and so the determinant of the $k\times k$ submatrix formed by the columns of $G_{2}$ indexed by $\mathcal{I}$ is zero. Therefore, $\mathcal{C}_{2}$ is not MRD.
\end{proof}

In the following, we discuss a class of twisted Gabidulin codes with $(t_{1},t_2)=(0,1)$.

Suppose that $t_{1}=0$ and $t_{2}=1$ in (\ref{p2}), we obtain
\begin{equation}\label{stp2}
\mathcal{P}_{2}=\left\{f(x)=\sum_{i=0}^{k-1} f_i x^{[i]}+\eta_{1}f_{h} x^{[k]}+\eta_{2}f_{h} x^{[k+1]}: f_i \in \mathbb{F}_{q^m},0\leq i\leq k-1\right\},
\end{equation}
where $0\leq h\leq k-1$.
In this case, %according to Theorem \ref{thm2},
we have
$$\begin{aligned}
1+\eta_{1} g_{h}^{(t_{1})}+\eta_{2} g_{h}^{(t_{2})}=&1+\eta_{1} g_{h}^{(0)}+\eta_{2} g_{h}^{(1)}=1-\eta_{1}c_{k-h}+\eta_{2}(c_{1}^{[1]}c_{k-h}-c_{k-h+1}^{[1]}).
\end{aligned}$$
Then we get the following corollary.
\begin{corollary}\label{t00}
Let $\mathcal{P}_{2}$ be as given in (\ref{stp2}), and let $\boldsymbol{\alpha}=[\alpha_{1},\ldots,\alpha_{n}]\in \mathbb{F}_{q^{m}}^{n}$ with $\mathrm{rk}_{\mathbb{F}_{q}}(\boldsymbol{\alpha})=n$.  For any $k$-subset $\mathcal{I}\subseteq\textit{\textbf{[}}n\textit{\textbf{]}}$, let $U_{\mathcal{I}}=\left\langle\alpha_{i} \ (i\in \mathcal{I}) \right\rangle_{\mathbb{F}_{q}}$,
\begin{equation}\label{uuu}
\prod_{\alpha \in U_{\mathcal{I}}}(x-\alpha)=\sum_{j=0}^{k} c_{j} x^{[k-j]},
\end{equation}
and $c_{k+1}=0$. If $\boldsymbol{\eta}=[\eta_{1}, \eta_{2}] \in \Omega_2^{\prime}$, where $\Omega_{2}^{\prime}=$
$$\left\{\boldsymbol{\eta}\in (\mathbb{F}_{q^{m}}^{*})^{2}: \exists \ k\text{-subset}\ \mathcal{I} \text{ such that } (\ref{uuu}) \text{ holds},  (\eta_{1}-\eta_{2}c_{1}^{[1]})c_{k-h}+\eta_{2}c_{k-h+1}^{[1]}= 1 \right\},$$
then $\mathcal{C}_{2}=\mathrm{ev}_{\boldsymbol{\alpha}}(\mathcal{P}_{2})$ is not MRD.
\end{corollary}

Below we consider the twisted Gabidulin codes $\mathcal{C}_{2}$ in Hamming metric, and we establish the necessary and sufficient conditions for $\mathcal{C}_{2}$ to be MDS, AMDS or NMDS.
\begin{theorem}
 $\mathcal{C}_{2}$ is MDS if and only if $\boldsymbol{\eta}=[\eta_{1}, \eta_{2}] \notin\Omega_2$.
\end{theorem}
\begin{proof}
Let $G_{2}$, as defined in (\ref{eqgt2}), be a generator matrix of $\mathcal{C}_{2}$.
 By Proposition \ref{gmds},
$\mathcal{C}_{2}$  is  MDS  if and only if for each  $k$-subset $\mathcal{I}\subseteq\textit{\textbf{[}}n\textit{\textbf{]}}$, the $k\times k$ minor formed by the columns of $G_{2}$ indexed by $\mathcal{I}$ is non-zero ($1+\eta_{1} g_{h}^{(t_{1})}(\mathcal{I})+\eta_{2} g_{h}^{(t_{2})}(\mathcal{I})\neq 0$), and if and only if $\boldsymbol{\eta}=[\eta_{1}, \eta_{2}] \notin\Omega_2$. It completes the proof.
\end{proof}

\begin{theorem}%\label{amds1}
Let $\boldsymbol{\eta}=[\eta_{1}, \eta_{2}]  \in \Omega_2$. Then $\mathcal{C}_2$  is AMDS if and only if for each $(k+1)$-subset  $\mathcal{J} \subseteq\textit{\textbf{[}}n\textit{\textbf{]}}$, there exists a $k$-subset  $\mathcal{I} \subseteq \mathcal{J}$ such that $1+\eta_{1} g_{h}^{(t_{1})}(\mathcal{I})+\eta_{2} g_{h}^{(t_{2})}(\mathcal{I})\neq 0$.
\end{theorem}
\begin{proof}
Since $\boldsymbol{\eta}=[\eta_{1}, \eta_{2}]  \in \Omega_2$, the condition (ii) of Lemma \ref{anmds} holds. Thus, the code $\mathcal{C}_2$ is AMDS if and only if the condition (iii) of Lemma \ref{anmds} holds,  if and only if for each $(k+1)$-subset  $\mathcal{J} \subseteq\textit{\textbf{[}}n\textit{\textbf{]}}$, there exists a $k$-subset  $\mathcal{I} \subseteq \mathcal{J}$ such that $1+\eta_{1} g_{h}^{(t_{1})}(\mathcal{I})+\eta_{2} g_{h}^{(t_{2})}(\mathcal{I})\neq 0$.
\end{proof}
\begin{theorem}%\label{nmds1}
Let $\boldsymbol{\eta}=[\eta_{1}, \eta_{2}]  \in \Omega_2$ and $h\in \{0,k-1\}$. Then $\mathcal{C}_2$ is NMDS if and only if for each $(k+1)$-subset  $\mathcal{J} \subseteq\textit{\textbf{[}}n\textit{\textbf{]}}$, there exists a $k$-subset  $\mathcal{I} \subseteq \mathcal{J}$ such that $1+\eta_{1} g_{h}^{(t_{1})}(\mathcal{I})+\eta_{2} g_{h}^{(t_{2})}(\mathcal{I})\neq 0$.
\end{theorem}
\begin{proof}
Since $h\in \{0,k-1\}$, the condition (i) of Lemma \ref{anmds} holds.
Since $\boldsymbol{\eta}=[\eta_{1}, \eta_{2}]  \in \Omega_2$, the condition (ii) of Lemma \ref{anmds} holds. Therefore, $\mathcal{C}_2$ is NMDS if and only if the condition (iii) of Lemma \ref{anmds} holds, if and only if for each $(k+1)$-subset  $\mathcal{J} \subseteq\textit{\textbf{[}}n\textit{\textbf{]}}$, there exists a $k$-subset  $\mathcal{I} \subseteq \mathcal{J}$ such that $1+\eta_{1} g_{h}^{(t_{1})}(\mathcal{I})+\eta_{2} g_{h}^{(t_{2})}(\mathcal{I})\neq 0$.
\end{proof}

\section{Twisted Gabidulin codes with $\ell$ twists}
By adding $\ell$ monomials to the same position of each $f(x)\in \mathbb{F}_{q^m}[x;\sigma]$ in the definition of Gabidulin codes, we can obtain twisted Gabidulin codes with $\ell$ twists which are defined as follows.
\begin{definition}\label{stl3}
Let $n, k, \ell$ be positive integers with $k<n\leq m$ and $2<\ell\leq n-k$. Let  $\boldsymbol{\eta}=[\eta_{1},\ldots,\eta_{\ell}]\in (\mathbb{F}_{q^m}^{\ast})^{\ell}$. Denote $\mathcal{P}_{3}$ as the set of twisted linearized polynomials over $\mathbb{F}_{q^{m}}$ given by
\begin{equation}\label{p3}
\mathcal{P}_{3}:=\left\{f(x)=\sum_{i=0}^{k-1} f_i x^{[i]}+f_{h}\sum_{j=1}^{\ell}\eta_{j}x^{[k+t_{j}]}: f_i \in \mathbb{F}_{q^m}, 0\leq i\leq k-1\right\},
\end{equation}
where $0\leq h\leq k-1$ and $0\leq t_{1}<\ldots<t_{\ell}\leq n-k-1$.  Let $\boldsymbol{\alpha}=[\alpha_{1},\ldots,\alpha_{n}]\in \mathbb{F}_{q^{m}}^{n}$ with $\mathrm{rk}_{\mathbb{F}_{q}}(\boldsymbol{\alpha})=n$. The twisted Gabidulin code with $\ell$ twists is defined by
\begin{equation}\label{c3}
\mathcal{C}_{3}:=\mathrm{ev}_{\boldsymbol{\alpha}}\left(\mathcal{P}_{3}\right)\subseteq \mathbb{F}_{q^m}^n.
\end{equation}
\end{definition}
Unless otherwise stated, let $\mathcal{C}_{3}$ be the twisted Gabidulin code defined as in (\ref{c3}). We first give two useful propositions about the code $\mathcal{C}_{3}$.

\begin{proposition}
$\mathcal{C}_{3}$ is an $[n, k]$ linear code over $\mathbb{F}_{q^{m}}$ with the generator matrix
\begin{equation}\label{eqgt3}
G_{3}=\left[\begin{array}{ccc}
\alpha_{1} & \cdots & \alpha_{n} \\
\vdots &  & \vdots \\
\alpha_{1}^{[h-1]} & \cdots & \alpha_{n}^{[h-1]} \\
\alpha_{1}^{[h]}+\sum_{j=1}^{\ell}\eta_{j}\alpha_{1}^{[k+t_{j}]} & \cdots & \alpha_{n}^{[h]}+\sum_{j=1}^{\ell}\eta_{j}\alpha_{n}^{[k+t_{j}]} \\
\alpha_{1}^{[h+1]} & \cdots & \alpha_{n}^{[h+1]} \\
\vdots  & & \vdots \\
\alpha_{1}^{[k-1]} & \cdots & \alpha_{n}^{[k-1]}
\end{array}\right].\end{equation}
\end{proposition}
\begin{proposition}
For the twisted Gabidulin code $\mathcal{C}_{3}$, we have $$\mathcal{C}_{3}=\left\langle\boldsymbol{\alpha}^{[h]}+\sum_{i=1}^{\ell}\eta_{i} \boldsymbol{\alpha}^{[k+t_{i}]}, \boldsymbol{\alpha}^{[s]}(s\in \{0,\ldots,k-1\}\backslash\{h\})  \right\rangle_{\mathbb{F}_{q^m}}.$$
\end{proposition}

Now we give the necessary and sufficient condition for the twisted Gabidulin code $\mathcal{C}_{3}$ to be MRD.

\begin{theorem}\label{mm3}
 $\mathcal{C}_3$ is MRD if and only if $\boldsymbol{\eta}=[\eta_{1}, \ldots,\eta_{\ell}] \in \mathcal{N}^{\prime}$, where
$$
\mathcal{N}^{\prime}=\left\{\boldsymbol{\eta} \in (\mathbb{F}_{q^{m}}^{*})^{\ell}: \forall V\in \mathcal{V}_q(k,n), \left|V M_{k}(\boldsymbol{\alpha})^{\top}\right|+\sum_{i=1}^{\ell}\eta_i\left|VM_{k}^{(h,k+t_i)}(\boldsymbol{\alpha})^{\top}\right| \neq 0\right\}.
$$
\end{theorem}
\begin{proof}
Let the matrix $G_3$ given in (\ref{eqgt3}) be a generator matrix of the twisted Gabidulin code  $\mathcal{C}_3$.

For any  $V\in \mathcal{V}_{q}(k, n)$, we have
$$\begin{aligned}
\left|V G_3^{\top}\right|&=\left|V\left[\begin{array}{cccc}
\boldsymbol{\alpha} \\
\vdots  \\
\boldsymbol{\alpha}^{[h-1]} \\
\boldsymbol{\alpha}^{[h]} \\
\boldsymbol{\alpha}^{[h+1]}  \\
\vdots \\
\boldsymbol{\alpha}^{[k-1]}
\end{array}\right]^{\top}\right|+\left|V\left[\begin{array}{cccc}
\boldsymbol{\alpha} \\
\vdots  \\
\boldsymbol{\alpha}^{[h-1]} \\
\eta_1 \boldsymbol{\alpha}^{[k+t_1]} \\
\boldsymbol{\alpha}^{[h+1]}  \\
\vdots \\
\boldsymbol{\alpha}^{[k-1]}
\end{array}\right]^{\top}\right|+\ldots+\left|V\left[\begin{array}{cccc}
\boldsymbol{\alpha} \\
\vdots  \\
\boldsymbol{\alpha}^{[h-1]} \\
\eta_{\ell} \boldsymbol{\alpha}^{[k+t_{\ell}]} \\
\boldsymbol{\alpha}^{[h+1]}  \\
\vdots \\
\boldsymbol{\alpha}^{[k-1]}
\end{array}\right]^{\top}\right|\\
&=\left|V M_{k}(\boldsymbol{\alpha})^{\top}\right|+\sum_{i=1}^{\ell}\eta_i\left|VM_{k}^{(h,k+t_i)}(\boldsymbol{\alpha})^{\top}\right|
\end{aligned}$$
Thus, by Proposition \ref{gmrd}, $\mathcal{C}_{3}$ is MRD if and only if $$\left|V M_{k}(\boldsymbol{\alpha})^{\top}\right|+\sum_{i=1}^{\ell}\eta_i\left|VM_{k}^{(h,k+t_i)}(\boldsymbol{\alpha})^{\top}\right| \neq 0,$$
that is,
$[\eta_{1}, \ldots,\eta_{\ell}] \in \mathcal{N}^{\prime}$.
\end{proof}
The following proposition shows a sufficient condition for the code $\mathcal{C}_{3}$ to be MRD, which is a special case of \cite{furt}. We also give a more direct proof using Theorem \ref{mm3}.
\begin{proposition}\label{spf1}
Let $s_1, \ldots,s_{\ell}$ be  positive integers with $n\leq s_1$ such that $\mathbb{F}_{q} \subsetneq \mathbb{F}_{q^{s_1}}\subsetneq \ldots \subsetneq \mathbb{F}_{q^{s_{\ell}}} \subsetneq \mathbb{F}_{q^{s_{\ell+1}}}=\mathbb{F}_{q^m}$ is a chain of subfields. Let $\alpha_{1}, \ldots, \alpha_{n}\in \mathbb{F}_{q^{s_1}}$ be linearly independent over $\mathbb{F}_{q}$ and  $\eta_{i}\in \mathbb{F}_{q^{s_{i+1}}} \backslash \mathbb{F}_{q^{s_{i}}}$ for $i=1,\ldots,\ell$. Then  $\mathcal{C}_{3}$  is MRD.
\end{proposition}
\begin{proof}
By Theorem \ref{mm3}, $\mathcal{C}_{3}$ is MRD if and only if $$\sum_{i=1}^{\ell}\eta_i\left|VM_{k}^{(h,k+t_i)}(\boldsymbol{\alpha})^{\top}\right| \neq -\left|V M_{k}(\boldsymbol{\alpha})^{\top}\right|$$ for any $V\in \mathcal{V}_{q}(k, n)$. Since $\alpha_{1}, \ldots, \alpha_{n}\in \mathbb{F}_{q^{s_1}}$, we have $\left|V M_{k}(\boldsymbol{\alpha})^{\top}\right|\in \mathbb{F}_{q^{s_{1}}}^*$. Denote $$A_{\ell}:=\sum_{i=1}^{\ell}\eta_i\left|VM_{k}^{(h,k+t_i)}(\boldsymbol{\alpha})^{\top}\right|=a_1\eta_1+\ldots+a_{\ell-1}\eta_{\ell-1}+a_{\ell}\eta_{\ell},$$
where $a_i=\left|VM_{k}^{(h,k+t_i)}(\boldsymbol{\alpha})^{\top}\right|\in\mathbb{F}_{q^{s_{1}}}$ for $1\leq i\leq \ell$.
Next we will prove $A_{\ell} \notin \mathbb{F}_{q^{s_{1}}}^*$ by induction on $\ell$.

 If $\ell=1$, then $A_1=a_1\eta_1 \notin \mathbb{F}_{q^{s_{1}}}^*$ since $\eta_{1}\in \mathbb{F}_{q^{s_{2}}} \backslash \mathbb{F}_{q^{s_{1}}}$. Suppose $A_{\ell-1}\notin \mathbb{F}_{q^{s_{1}}}^*$. %Now we discuss the value of $S_{\ell}$.
 If $a_{\ell}=0$, then $A_{\ell}=A_{\ell-1}\notin \mathbb{F}_{q^{s_{1}}}^*$. Otherwise,  we have
$$\eta_{\ell}= \frac{A_{\ell}-A_{\ell-1}}{a_{\ell}}\notin \mathbb{F}_{q^{s_{\ell}}}.$$
Thus, $A_{\ell}\notin \mathbb{F}_{q^{s_{\ell}}}$. In particular, $A_{\ell}\notin \mathbb{F}_{q^{s_{1}}}^*$.

From $\left|V M_{k}(\boldsymbol{\alpha})^{\top}\right|\in \mathbb{F}_{q^{s_{1}}}^*$ and $A_{\ell}\notin \mathbb{F}_{q^{s_{1}}}^*$, we obtain $A_{\ell}\neq -\left|V M_{k}(\boldsymbol{\alpha})^{\top}\right|$. Therefore, $\mathcal{C}_{3}$ is MRD. \end{proof}

Next, we construct two new classes of MRD twisted Gabidulin codes.

\begin{proposition}\label{spf0}
Let $s$ be a positive integer with $n\leq s$ such that $\mathbb{F}_{q} \subsetneq \mathbb{F}_{q^{s}} \subsetneq \mathbb{F}_{q^m}$ is a chain of subfields. Let $\alpha_{1}, \ldots, \alpha_{n}\in \mathbb{F}_{q^s}$ be linearly independent over $\mathbb{F}_{q}$ and  $\boldsymbol{\eta}=[\eta_1,\ldots,\eta_{\ell}]\in\mathbb{F}_{q^m}^{\ell}$ with $\eta_1\in \mathbb{F}_{q^m} \backslash \mathbb{F}_{q^s}$, where $\eta_i=b_i\eta_1$ and $b_i\in \mathbb{F}_{q^s}^*$ for $1<i\leq \ell$. Then  $\mathcal{C}_{3}$  is MRD.
\end{proposition}
\begin{proof}
%Let $G_3$ given in (\ref{eqgt3}) be a generator matrix of the code  $\mathcal{C}_3$.
For any $V\in \mathcal{V}_{q}(k, n)$, since $\alpha_{1}, \ldots, \alpha_{n}\in \mathbb{F}_{q^{s}}$, we have $\left|V M_{k}(\boldsymbol{\alpha})^{\top}\right|\in \mathbb{F}_{q^{s}}^*$. For $\eta_1\in \mathbb{F}_{q^m} \backslash \mathbb{F}_{q^s}$, we obtain
$$\sum_{i=1}^{\ell}\eta_i\left|VM_{k}^{(h,k+t_i)}(\boldsymbol{\alpha})^{\top}\right|= \eta_1\left(\left|VM_{k}^{(h,k+t_1)}(\boldsymbol{\alpha})^{\top}\right|+\sum_{i=2}^{\ell}b_i\left|VM_{k}^{(h,k+t_i)}(\boldsymbol{\alpha})^{\top}\right|\right) \notin \mathbb{F}_{q^{s}}^*$$
since $\eta_i=b_i\eta_1$ and $b_i\in \mathbb{F}_{q^s}^*$ for $1<i\leq \ell$.
Hence, $\sum_{i=1}^{\ell}\eta_i\left|VM_{k}^{(h,k+t_i)}(\boldsymbol{\alpha})^{\top}\right| \neq -\left|V M_{k}(\boldsymbol{\alpha})^{\top}\right|.$
By Theorem \ref{mm3}, $\mathcal{C}_{3}$ is MRD.
\end{proof}

\begin{definition}[\cite{tspf}]
Let  $\mathbb{F}_{q^m} / \mathbb{F}_{q^s}$  be a field extension. A vector  $\boldsymbol{\eta} \in   \mathbb{F}_{q^m}^{\ell}$  is called  $t$-$\mathbb{F}_{q^s}$-sum-product free if
$$\sum_{\substack{\mathcal{S} \subseteq\{1, \ldots, \ell\} \\ 0<|\mathcal{S}|\leq t}} a_{\mathcal{S}} \prod_{i \in \mathcal{S}} \eta_{i} \notin \mathbb{F}_{q^s}^{*}, \quad \forall a_{\mathcal{S}} \in \mathbb{F}_{q^s}.$$
\end{definition}

\begin{theorem}%\label{pro45}
Let $s$ be a positive integer with $n\leq s$ such that $\mathbb{F}_{q} \subsetneq \mathbb{F}_{q^{s}} \subsetneq \mathbb{F}_{q^m}$ is a chain of subfields. Let  $\alpha_{1},\ldots,\alpha_{n}\in \mathbb{F}_{q^{s}}$ be linearly independent over $\mathbb{F}_{q}$. Let $\boldsymbol{\eta}=[\eta_1,\ldots,\eta_\ell]\in \mathbb{F}_{q^{m}}^{\ell}$ be $1$-$\mathbb{F}_{q^{s}}$-sum-product free, then the twisted Gabidulin code $\mathcal{C}_{3}$ is MRD.
\end{theorem}
\begin{proof} 
For any $V\in \mathcal{V}_{q}(k, n)$, since each $\alpha_i\in \mathbb{F}_{q^s}$, we obtain
 $$\left|V M_{k}(\boldsymbol{\alpha})^{\top}\right|\in \mathbb{F}_{q^s}^{*} \ \text{  and }  \left|VM_{k}^{(h,k+t_i)}(\boldsymbol{\alpha})^{\top}\right|\in\mathbb{F}_{q^s}, 1\leq i\leq \ell.$$ Since $\boldsymbol{\eta}\in \mathbb{F}_{q^{m}}^{\ell}$ is $1$-$\mathbb{F}_{q^{s}}$-sum-product free, we have $\sum_{i=1}^{\ell}\eta_i\left|VM_{k}^{(h,k+t_i)}(\boldsymbol{\alpha})^{\top}\right|\notin \mathbb{F}_{q^s}^{*}$, and so $\sum_{i=1}^{\ell}\eta_i\left|VM_{k}^{(h,k+t_i)}(\boldsymbol{\alpha})^{\top}\right| \neq -\left|V M_{k}(\boldsymbol{\alpha})^{\top}\right|$. Thus, by Theorem \ref{mm3}, the code $\mathcal{C}_{3}$ is MRD.
 \end{proof}

\begin{remark}
The above Proposition \ref{spf1} and Proposition \ref{spf0} can be seen as two constructions of $1$-$\mathbb{F}_{q^{s}}$-sum-product free sets.
\end{remark}

Below, we present a useful lemma that plays an important role.
\begin{lemma}\label{lemma3}
Suppose that $\alpha_{1},\ldots,\alpha_{k}\in \mathbb{F}_{q^m}$ are linearly independent over $\mathbb{F}_{q}$. Let $\prod_{\alpha \in \left\langle\alpha_{1},\ldots,\alpha_{k}  \right\rangle_{\mathbb{F}_{q}}}(x-\alpha)=\sum_{j=0}^{k} c_{j} x^{[k-j]}$ and $c_j=0$ for $j>k$. Let $G_{3k}$ be the submatrix formed by the first $k$ columns of $G_{3}$ given in (\ref{eqgt3}), then $|G_{3k}|$=0 if and only if %$\eta_{1},\ldots, \eta_{\ell}$ satisfy
$1+\sum_{j=1}^{\ell}\eta_{j} g_{h}^{(t_{j})}=0$, where $g_{h}^{(t_{j})}=-\sum_{i=0}^{\min \{t_{j}, h\}} e_{t_{j},i} c_{k-h+i}^{[i]}$ and $e_{t_{j},0},\ldots,e_{t_{j},t_{j}}$ are as defined in Lemma \ref{lem1} for $1\leq j\leq \ell$.
\end{lemma}
\begin{proof}
By computing the determinant of $G_{3k}$, and applying Lemma \ref{dk} and Lemma \ref{lem3}, we obtain
$$\begin{aligned}
& \ |G_{3k}|\\
=&\left|\begin{array}{ccc}
\alpha_{1} & \cdots & \alpha_{k} \\
\vdots &  & \vdots \\
\alpha_{1}^{[h-1]} & \cdots & \alpha_{k}^{[h-1]} \\
\alpha_{1}^{[h]}+\sum_{j=1}^{\ell}\eta_{j}\alpha_{1}^{[k+t_{j}]} & \cdots & \alpha_{k}^{[h]}+\sum_{j=1}^{\ell}\eta_{j}\alpha_{k}^{[k+t_{j}]} \\
\alpha_{1}^{[h+1]}  & \cdots & \alpha_{k}^{[h+1]} \\
\vdots  & & \vdots \\
\alpha_{1}^{[k-1]}  & \cdots & \alpha_{k}^{[k-1]}
\end{array}\right|\\
=&\left|\begin{array}{ccc}
\alpha_{1}  & \cdots & \alpha_{k} \\
\vdots  &  & \vdots \\
\alpha_{1}^{[h-1]}  & \cdots & \alpha_{k}^{[h-1]} \\
\alpha_{1}^{[h]}  & \cdots & \alpha_{k}^{[h]} \\
\alpha_{1}^{[h+1]}  & \cdots & \alpha_{k}^{[h+1]} \\
\vdots  & &\vdots \\
\alpha_{1}^{[k-1]} & \cdots & \alpha_{k}^{[k-1]}
\end{array}\right|+\eta_{1} \left|\begin{array}{ccc}
\alpha_{1}  & \cdots & \alpha_{k} \\
\vdots  &  & \vdots \\
\alpha_{1}^{[h-1]}  & \cdots & \alpha_{k}^{[h-1]} \\
\alpha_{1}^{[k+t_{1}]}  & \cdots & \alpha_{k}^{[k+t_{1}]} \\
\alpha_{1}^{[h+1]}  & \cdots & \alpha_{k}^{[h+1]} \\
\vdots  & & \vdots \\
\alpha_{1}^{[k-1]} & \cdots & \alpha_{k}^{[k-1]}
\end{array}\right|+\ldots+\eta_{\ell} \left|\begin{array}{ccc}
\alpha_{1}  & \cdots & \alpha_{k} \\
\vdots  &  & \vdots \\
\alpha_{1}^{[h-1]}  & \cdots & \alpha_{k}^{[h-1]} \\
\alpha_{1}^{[k+t_{\ell}]}  & \cdots & \alpha_{k}^{[k+t_{\ell}]} \\
\alpha_{1}^{[h+1]}  & \cdots & \alpha_{k}^{[h+1]} \\
\vdots  & & \vdots \\
\alpha_{1}^{[k-1]} & \cdots & \alpha_{k}^{[k-1]}
\end{array}\right|\\
=&(1+\sum_{j=1}^{\ell}\eta_{j} g_{h}^{(t_{j})})\cdot \left|M_{k}(\boldsymbol{\alpha})\right|.\end{aligned}$$
Then $|G_{3k}|$=0 if and only if $1+\sum_{j=1}^{\ell}\eta_{j} g_{h}^{(t_{j})}=0$. It completes the proof.
\end{proof}
Now we give the condition that $\mathcal{C}_{3}$ is not an MRD code.
\begin{theorem}\label{tno6}
Let $g_{h}^{(t_{j})}=-\sum_{i=0}^{\min \{t_{j}, h\}} e_{t_{j},i} c_{k-h+i}^{[i]}$ for $1\leq j\leq \ell$, where $e_{t_{j},0},\ldots,e_{t_{j},t_{j}}$ are as defined in Lemma \ref{lem1}, $c_{0},\ldots,c_{k}$ satisfy $\prod_{\alpha \in \left\langle\alpha_{1},\ldots,\alpha_{k}  \right\rangle_{\mathbb{F}_{q}}}(x-\alpha)=\sum_{j=0}^{k} c_{j} x^{[k-j]}$, and $c_j=0$ for $j>k$. If $1+\sum_{j=1}^{\ell}\eta_{j} g_{h}^{(t_{j})}=0$, then $\mathcal{C}_{3}$  is not MRD.
\end{theorem}
\begin{proof}
Let $G_{3}$ given in (\ref{eqgt3}) be a generator matrix of $\mathcal{C}_{3}$.
By Lemma \ref{lemma3}, if $1+\sum_{j=1}^{\ell}\eta_{j} g_{h}^{(t_{j})}=0$, then the $k\times k$ minor formed by the first $k$ columns of $G_{3}$ is zero. From Lemma \ref{nmrd},  $\mathcal{C}_{3}$  is not MRD.
\end{proof}

\begin{definition}
Suppose that $\alpha_{1},\ldots,\alpha_{n}\in \mathbb{F}_{q^m}$ are linearly independent over $\mathbb{F}_{q}$. For any $k$-subset $\mathcal{I}\subseteq\textit{\textbf{[}}n\textit{\textbf{]}}$, let $U_{\mathcal{I}}=\left\langle\alpha_{i} \ (i\in \mathcal{I}) \right\rangle_{\mathbb{F}_{q}}$, $\prod_{\alpha \in U_{\mathcal{I}}}(x-\alpha)=\sum_{j=0}^{k} c_{j} x^{[k-j]}$, and $c_j=0$ for $j>k$. Let $e_{t_{j},0},\ldots,e_{t_{j},t_{j}}$ be as defined in Lemma \ref{lem1} and $g_{h}^{(t_{j})}(\mathcal{I})=-\sum_{i=0}^{\min \{t_{j}, h\}} e_{t_{j},i} c_{k-h+i}^{[i]}$ for $1\leq j\leq \ell$. We denote
$$
\Omega_3:=\left\{[\eta_{1}, \ldots, \eta_{\ell}] \in (\mathbb{F}_{q^{m}}^{*})^{\ell}:\    1+\sum_{j=1}^{\ell}\eta_{j} g_{h}^{(t_{j})}(\mathcal{I})= 0\ \text{for some}\ k\text{-subset}\ \mathcal{I} \text{ of } \textit{\textbf{[}}n\textit{\textbf{]}}\right\}.
$$
\end{definition}
Using the same argument in the proofs of Lemma \ref{lemma3} and Theorem \ref{tno6}, we obtain the following theorem.
\begin{theorem}
If $\boldsymbol{\eta}=[\eta_{1}, \ldots,\eta_{\ell}] \in\Omega_3$, then $\mathcal{C}_{3}$  is not MRD.
\end{theorem}
 \begin{proof}
Let $G_{3}$ given in (\ref{eqgt3}) be a generator matrix of $\mathcal{C}_{3}$.
If $\boldsymbol{\eta}\in\Omega_3$, then there exists a $k$-subset $\mathcal{I}\subseteq\textit{\textbf{[}}n\textit{\textbf{]}}$ such that $1+\sum_{j=1}^{\ell}\eta_{j} g_{h}^{(t_{j})}(\mathcal{I})=0$, and so the determinant of the $k\times k$ submatrix formed by the columns of $G_{3}$ indexed by $\mathcal{I}$ is zero. Thus, $\mathcal{C}_{3}$ is not MRD.
\end{proof}

Considering the code $\mathcal{C}_{3}$ in the Hamming metric, we obtain the following results.
\begin{theorem}
 $\mathcal{C}_{3}$ is MDS if and only if $\boldsymbol{\eta}=[\eta_{1}, \ldots\eta_{\ell}] \notin\Omega_3$.
\end{theorem}
\begin{proof}
Let $G_{3}$, as defined in (\ref{eqgt3}), be a generator matrix of $\mathcal{C}_{3}$.
 By Proposition \ref{gmds},
$\mathcal{C}_{3}$  is  MDS  if and only if for each  $k$-subset $\mathcal{I}\subseteq\textit{\textbf{[}}n\textit{\textbf{]}}$, the $k\times k$ minor formed by the columns of $G_{3}$ indexed by $\mathcal{I}$ is non-zero ($1+\sum_{j=1}^{\ell}\eta_{j} g_{h}^{(t_{j})}(\mathcal{I})\neq 0$), and if and only if $\boldsymbol{\eta}=[\eta_{1}, \ldots\eta_{\ell}]  \notin\Omega_3$. It completes the proof.
\end{proof}

\begin{theorem}
Let $\boldsymbol{\eta}=[\eta_{1}, \ldots\eta_{\ell}] \in \Omega_3$. Then $\mathcal{C}_3$ is AMDS if and only if for each $(k+1)$-subset $\mathcal{J} \subseteq\textit{\textbf{[}}n\textit{\textbf{]}}$, there exists a $k$-subset  $\mathcal{I} \subseteq \mathcal{J}$ such that $1+\sum_{j=1}^{\ell}\eta_{j} g_{h}^{(t_{j})}(\mathcal{I})\neq 0$.
\end{theorem}
\begin{proof}
Since $\boldsymbol{\eta}=[\eta_{1}, \ldots\eta_{\ell}] \in \Omega_3$, the condition (ii) of Lemma \ref{anmds} holds. Thus, the code $\mathcal{C}_3$ is AMDS if and only if the condition (iii) of Lemma \ref{anmds} holds,  if and only if for each $(k+1)$-subset  $\mathcal{J} \subseteq\textit{\textbf{[}}n\textit{\textbf{]}}$, there exists a $k$-subset  $\mathcal{I} \subseteq \mathcal{J}$ such that $1+\sum_{j=1}^{\ell}\eta_{j} g_{h}^{(t_{j})}(\mathcal{I})\neq 0$.
\end{proof}
\begin{theorem}
Let $\boldsymbol{\eta}=[\eta_{1}, \ldots\eta_{\ell}] \in \Omega_3$ and $h\in \{0,k-1\}$. Then $\mathcal{C}_3$ is NMDS if and only if for each $(k+1)$-subset  $\mathcal{J} \subseteq\textit{\textbf{[}}n\textit{\textbf{]}}$, there exists a $k$-subset  $\mathcal{I} \subseteq \mathcal{J}$ such that $1+\sum_{j=1}^{\ell}\eta_{j} g_{h}^{(t_{j})}(\mathcal{I})\neq 0$.
\end{theorem}
\begin{proof}
Since $h\in \{0,k-1\}$, the condition (i) of Lemma \ref{anmds} holds.
Since $\boldsymbol{\eta}=[\eta_{1}, \ldots\eta_{\ell}] \in \Omega_3$, the condition (ii) of Lemma \ref{anmds} holds. Therefore, $\mathcal{C}_3$ is NMDS if and only if the condition (iii) of Lemma \ref{anmds} holds, if and only if for each $(k+1)$-subset  $\mathcal{J} \subseteq\textit{\textbf{[}}n\textit{\textbf{]}}$, there exists a $k$-subset  $\mathcal{I} \subseteq \mathcal{J}$ such that $1+\sum_{j=1}^{\ell}\eta_{j} g_{h}^{(t_{j})}(\mathcal{I})\neq 0$.
\end{proof}

\section{Covering radii and deep holes of  twisted Gabidulin codes}
In this section, we study the covering radii of three classes of twisted Gabidulin codes, and then we characterize the deep holes for certain twisted Gabidulin codes.  We first recall some basic concepts and well-known results.

Let $\mathcal{C}$ be an $[n,k]$ rank metric code over $\mathbb{F}_{q^m}$. For any vector $\boldsymbol{u}\in\mathbb{F}_{q^m}^n$, the rank distance from $\boldsymbol{u}$ to the code $\mathcal{C}$ is defined as $d_R(\boldsymbol{u},\mathcal{C}):=\mathrm{min}\{d_R(\boldsymbol{u},\boldsymbol{c}):\boldsymbol{c}\in\mathcal{C}\}$.
According to this, the covering radius and deep hole of a rank metric code are defined as follows.
\begin{definition}[\cite{st}]\label{crdh}
Let $\mathcal{C}\subseteq \mathbb{F}_{q^m}^{n}$ be an $[n,k]$ rank metric code. The covering radius of $\mathcal{C}$ is defined as
$$\rho_{R}(\mathcal{C}):=\mathrm{max}\{d_R(\boldsymbol{u},\mathcal{C}):\boldsymbol{u}\in\mathbb{F}_{q^m}^n\}.$$
A vector $\boldsymbol{x}\in \mathbb{F}_{q^m}^{n}$ is called a deep hole of $\mathcal{C}$ if it satisfies
$$d_R(\boldsymbol{x},\mathcal{C})=\rho_{R}(\mathcal{C}).$$
\end{definition}
The following states an important fact.
\begin{lemma}[\cite{st}]\label{cr}
Let $\mathcal{C},\mathcal{D}\subseteq \mathbb{F}_{q^m}^n$ be a pair of linear rank metric codes. The following hold.
\begin{enumerate}[(i)]
\item If $\mathcal{C}\subseteq\mathcal{D}$, then $\rho_{R}(\mathcal{C})\geq\rho_{R}(\mathcal{D})$.
\item If $\mathcal{C}\subsetneq\mathcal{D}$, then $\rho_{R}(\mathcal{C})\geq d_R(\mathcal{D})$.
\item $\rho_{R}(\mathcal{C})=0$ if and only if $\mathcal{C}=\mathbb{F}_{q^m}^n$.
\end{enumerate}
\end{lemma}

Note that a linear code with Hamming metric covering radius $\rho_H$ has rank metric covering radius $\rho_{R}\leq \rho_H$ \cite{packpro} since for two vectors $\boldsymbol{u},\boldsymbol{c}\in\mathbb{F}_{q^m}^n$, the rank distance between them is less than or equal to the Hamming distance between them, i.e., $ d_R(\boldsymbol{u},\boldsymbol{c})\leq d_H(\boldsymbol{u},\boldsymbol{c})$.
From the redundancy bound \cite[Corollary 11.1.3]{dp2},  we can obtain the following useful lemma.
\begin{lemma}\label{cra}
Let $\mathcal{C}$ be an $[n,k]$ linear rank metric code over $\mathbb{F}_{q^m}$ with $k< n\leq m$. It holds that $\rho_{R}(\mathcal{C})\leq n-k$.
\end{lemma}
In the following, we consider the twisted Gabidulin codes $\mathcal{C}_{1}=\mathrm{ev}_{\boldsymbol{\alpha}}(\mathcal{P}_{1})$ with $t=0$, where $\mathcal{P}_{1}$ is given in (\ref{stp1}).
Firstly, we determine the covering radius of $\mathcal{C}_{1}$.
\begin{theorem}\label{crho1} Let $\mathcal{C}_{1}=\mathrm{ev}_{\boldsymbol{\alpha}}(\mathcal{P}_{1})$ be an $[n,k]$ twisted Gabidulin code, where $\mathcal{P}_{1}$ is given in (\ref{stp1}) and $\boldsymbol{\alpha}=[\alpha_{1},\ldots,\alpha_{n}]\in \mathbb{F}_{q^{m}}^{n}$ with $\mathrm{rk}_{\mathbb{F}_{q}}(\boldsymbol{\alpha})=n$. It holds that
$$\rho_{R}(\mathcal{C}_1)=n-k.$$
\end{theorem}
\begin{proof}
Since $\mathcal{C}_1$ has generator matrix
\begin{equation}\label{cdc1}
\left[\begin{array}{cccc}
\alpha_{1} & \alpha_{2} & \ldots & \alpha_{n} \\
\vdots & \vdots & & \vdots \\
\alpha_{1}^{[h]}+\eta \alpha_{1}^{[k]} & \alpha_{2}^{[h]}+\eta\alpha_{2}^{[k]} & \ldots & \alpha_{n}^{[h]}+\eta\alpha_{n}^{[k]} \\
\vdots & \vdots &  & \vdots \\
\alpha_{1}^{[k-1]} & \alpha_{2}^{[k-1]} & \ldots & \alpha_{n}^{[k-1]}
\end{array}\right],
\end{equation}
we can easily obtain $\mathcal{C}_1\subsetneq \mathcal{C}$, where $\mathcal{C}$ is an $[n,k+1]$ linear code with a generator matrix of the form
$$
\left[\begin{array}{cccc}
\alpha_{1} & \alpha_{2} & \ldots & \alpha_{n} \\
\vdots & \vdots & & \vdots \\
\alpha_{1}^{[k-1]} & \alpha_{2}^{[k-1]} & \ldots & \alpha_{n}^{[k-1]}\\
\alpha_{1}^{[k]} & \alpha_{2}^{[k]} & \ldots & \alpha_{n}^{[k]}
\end{array}\right].
$$
Obviously, $\mathcal{C}$ is the Gabidulin code $\mathcal{G}_{n,k+1}(\boldsymbol{\alpha})$, then we have $d_R(\mathcal{C})=n-k$. By Lemma \ref{cr}, we obtain $\rho_{R}(\mathcal{C}_1)\geq d_R(\mathcal{C})=n-k$. It follows from Lemma \ref{cra} that $\rho_{R}(\mathcal{C}_1)\leq n-k$. Therefore, $\rho_{R}(\mathcal{C}_1)=n-k$ holds.
\end{proof}
Next, we provide an equivalent condition under which a vector is a deep hole of the code $\mathcal{C}_1$.

\begin{theorem}\label{dh1}
Let $G_1$ given in (\ref{cdc1}) be a generator matrix of the twisted Gabidulin code $\mathcal{C}_1$. For any vector $\boldsymbol{u}\in\mathbb{F}_{q^m}^n\backslash \mathcal{C}_1$, let $G=\begin{bmatrix}
G_1 \\
\boldsymbol{u}
\end{bmatrix}$. Then $\boldsymbol{u}$ is a deep hole of $\mathcal{C}_1$ if and only if $G$ generates an $[n,k+1]$ MRD code $\mathcal{C}$.
\end{theorem}
\begin{proof}
Note that $$d_R(\mathcal{C})=\mathrm{min}\{d_R(\mathcal{C}_1),d_R(\boldsymbol{u},\mathcal{C}_1)\}.$$
From the proof of Theorem \ref{crho1}, since $\mathcal{C}_1\subsetneq \mathcal{G}_{n,k+1}(\boldsymbol{\alpha})$, we have
\begin{equation}\label{en1}
d_R(\mathcal{C}_1)\geq d_R(\mathcal{G}_{n,k+1}(\boldsymbol{\alpha}))=n-k.
\end{equation}
By Theorem \ref{crho1} and Definition \ref{crdh}, we have
\begin{equation}\label{en2}
\rho_{R}(\mathcal{C}_1)=n-k\geq d_R(\boldsymbol{u},\mathcal{C}_1).
\end{equation}
It follows from (\ref{en1}) and \eqref{en2} that $d_R(\mathcal{C}_1)\geq d_R(\boldsymbol{u},\mathcal{C}_1)$.
Hence, we immediately obtain  $$d_R(\mathcal{C})=d_R(\boldsymbol{u},\mathcal{C}_1).$$
 Therefore, $\boldsymbol{u}$ is a deep hole of $\mathcal{C}_1$ if and only if $d_R(\boldsymbol{u},\mathcal{C}_1)=\rho_{R}(\mathcal{C}_1)=n-k$, if and only if $\mathcal{C}$ generated by $G$ is an $[n,k+1]$ MRD code. It completes the proof.
\end{proof}

According to Theorem \ref{dh1}, we can obtain two classes of deep holes of the twisted Gabidulin code $\mathcal{C}_1$ with $t=0$.
\begin{theorem}\label{dh2}
Let $\mathcal{C}_{1}=\mathrm{ev}_{\boldsymbol{\alpha}}(\mathcal{P}_{1})$ be an $[n,k]$ twisted Gabidulin code, where $\mathcal{P}_{1}$ is given in (\ref{stp1}) and $\boldsymbol{\alpha}=[\alpha_{1},\ldots,\alpha_{n}]\in \mathbb{F}_{q^{m}}^{n}$ with $\mathrm{rk}_{\mathbb{F}_{q}}(\boldsymbol{\alpha})=n$. Let $g(x)=gx^{[k]}+f(x)$, where $g\in\mathbb{F}^{*}_{q^m}$ and $f(x)\in P_{1}$. Let $\boldsymbol{u}_{g}=(g(\alpha_1),\ldots,g(\alpha_n))\in \mathbb{F}^{n}_{q^m}$. Then $\boldsymbol{u}_g$ is a deep hole of $\mathcal{C}_1$.
\end{theorem}
\begin{proof}
Let the matrix $G_1$ given in (\ref{cdc1}) be a generator matrix of $\mathcal{C}_1$. Then we have
$$\begin{bmatrix}
G_1 \\
\boldsymbol{u}_{g}
\end{bmatrix}=\left[\begin{array}{cccc}
\alpha_{1} & \alpha_{2} & \cdots & \alpha_{n} \\
\vdots &\vdots &  & \vdots \\
\alpha_{1}^{[h-1]} & \alpha_{2}^{[h-1]} & \cdots & \alpha_{n}^{[h-1]} \\
\alpha_{1}^{[h]}+\eta \alpha_{1}^{[k]} &\alpha_{2}^{[h]}+\eta \alpha_{2}^{[k]} & \cdots & \alpha_{n}^{[h]}+\eta \alpha_{n}^{[k]} \\
\alpha_{1}^{[h+1]} & \alpha_{2}^{[h+1]} & \cdots & \alpha_{n}^{[h+1]} \\
\vdots & \vdots & & \vdots \\
\alpha_{1}^{[k-1]} & \alpha_{2}^{[k-1]} & \cdots & \alpha_{n}^{[k-1]}\\
g\alpha_{1}^{[k]}+f(\alpha_{1}) & g\alpha_{2}^{[k]}+f(\alpha_{2}) & \cdots &g\alpha_{n}^{[k]}+f(\alpha_{n})
\end{array}\right],$$
where $f(\alpha_{j})=\sum_{i=0}^{k-1} f_i \alpha_{j}^{[i]}+\eta f_{h} \alpha_{j}^{[k]}$ for $1\leq j\leq n$. It is clear that $\begin{bmatrix}
G_1 \\
\boldsymbol{u}_{g}
\end{bmatrix}$ is row equivalent to
$$
\left[\begin{array}{cccc}
\alpha_{1} & \alpha_{2} & \cdots & \alpha_{n} \\
\vdots &\vdots &  & \vdots \\
\alpha_{1}^{[k-1]} & \alpha_{2}^{[k-1]} & \cdots & \alpha_{n}^{[k-1]}\\
\alpha_{1}^{[k]} & \alpha_{2}^{[k]} & \cdots & \alpha_{n}^{[k]}\\
\end{array}\right],$$
which generates the Gabidulin code $\mathcal{G}_{n,k+1}(\boldsymbol{\alpha})$. By Theorem \ref{dh1}, $\boldsymbol{u}_{g}$ is a deep hole of $\mathcal{C}_1$.
\end{proof}
\begin{theorem}
 Let $\mathcal{C}_{1}=\mathrm{ev}_{\boldsymbol{\alpha}}(\mathcal{P}_{1})$ be an $[n,k]$  twisted Gabidulin code,  where $\mathcal{P}_{1}$ is given in (\ref{stp1}) and $\boldsymbol{\alpha}=[\alpha_{1},\ldots,\alpha_{n}]\in \mathbb{F}_{q^{m}}^{n}$ with $\mathrm{rk}_{\mathbb{F}_{q}}(\boldsymbol{\alpha})=n$.  Let $g(x)=gx^{[h]}+f(x)$, where $g\in\mathbb{F}^{*}_{q^m}$ and $f(x)\in P_{1}$. Let $\boldsymbol{u}_{g}=(g(\alpha_1),\ldots,g(\alpha_n))\in \mathbb{F}^{n}_{q^m}$. Then $\boldsymbol{u}_g$ is a deep hole of $\mathcal{C}_1$.
\end{theorem}
\begin{proof}
The proof is similar to that of Theorem \ref{dh2}, thus we omit it.
\end{proof}
We can also characterize the covering radii of twisted Gabidulin codes with two twists and with $\ell$ twists.
\begin{theorem}\label{crt2} Let $\mathcal{C}_{2}=\mathrm{ev}_{\boldsymbol{\alpha}}(\mathcal{P}_{2})$ be an $[n,k]$ twisted Gabidulin code, where $\mathcal{P}_{2}$ is given in (\ref{stp2}) and $\boldsymbol{\alpha}=[\alpha_{1},\ldots,\alpha_{n}]\in \mathbb{F}_{q^{m}}^{n}$ with $\mathrm{rk}_{\mathbb{F}_{q}}(\boldsymbol{\alpha})=n$. It holds that
$$n-k-1\leq \rho_{R}(\mathcal{C}_2)\leq n-k.$$
\end{theorem}
\begin{proof}
Since  $\mathcal{C}_2$ has generator matrix
$$
\left[\begin{array}{cccc}
\alpha_{1} & \alpha_{2} & \ldots & \alpha_{n} \\
\vdots & \vdots &  & \vdots \\
\alpha_{1}^{[h]}+\eta_1 \alpha_{1}^{[k]}+\eta_2\alpha_{1}^{[k+1]} & \alpha_{2}^{[h]}+\eta_1\alpha_{2}^{[k]}+\eta_2\alpha_{2}^{[k+1]} & \ldots & \alpha_{n}^{[h]}+\eta_1\alpha_{n}^{[k]}+\eta_2\alpha_{n}^{[k+1]} \\
\vdots & \vdots &  & \vdots \\
\alpha_{1}^{[k-1]} & \alpha_{2}^{[k-1]} & \ldots & \alpha_{n}^{[k-1]}
\end{array}\right],
$$
we can easily obtain $\mathcal{C}_2\subsetneq \mathcal{C}$, where $\mathcal{C}$ is an $[n,k+2]$ linear code with a generator matrix of the form
$$
\left[\begin{array}{cccc}
\alpha_{1} & \alpha_{2} & \ldots & \alpha_{n} \\
\vdots & \vdots &  & \vdots \\
\alpha_{1}^{[k-1]} & \alpha_{2}^{[k-1]} & \ldots & \alpha_{n}^{[k-1]}\\
\alpha_{1}^{[k]} & \alpha_{2}^{[k]} & \ldots & \alpha_{n}^{[k]}\\
\alpha_{1}^{[k+1]} & \alpha_{2}^{[k+1]} & \ldots & \alpha_{n}^{[k+1]}
\end{array}\right].
$$
Since $\mathcal{C}$ is the  Gabidulin code $\mathcal{G}_{n,k+2}(\boldsymbol{\alpha})$, we have $d_R(\mathcal{C})=n-k-1$. By Lemma \ref{cr}, we obtain $\rho_{R}(\mathcal{C}_2)\geq d_R(\mathcal{C})=n-k-1$. It follows from Lemma \ref{cra} that $\rho_{R}(\mathcal{C}_2)\leq n-k$. Therefore, $n-k-1\leq \rho_{R}(\mathcal{C}_2)\leq n-k$ holds.
\end{proof}
\begin{theorem}
Let $\mathcal{C}_{3}=\mathrm{ev}_{\boldsymbol{\alpha}}(\mathcal{P}_{3})$ be an $[n,k]$ twisted Gabidulin code, where
 $$
 \mathcal{P}_{3}=\left\{f(x)=\sum_{i=0}^{k-1} f_i x^{[i]}+f_{h}\sum_{j=0}^{\ell-1}\eta_{j+1}x^{[k+j]}: f_i \in \mathbb{F}_{q^m},0\leq i\leq k-1 \right\}
 $$
 and $\boldsymbol{\alpha}=[\alpha_{1},\ldots,\alpha_{n}]\in \mathbb{F}_{q^{m}}^{n}$ with $\mathrm{rk}_{\mathbb{F}_{q}}(\boldsymbol{\alpha})=n$.
 It holds that
$$n-k-\ell+1 \leq \rho_{R}(\mathcal{C}_3)\leq n-k.$$
\end{theorem}
\begin{proof}
The proof is similar to that of Theorem \ref{crt2}, so we omit it.
\end{proof}

\section{Conclusion}
In this paper, we studied three classes of twisted Gabidulin codes in both the rank metric and Hamming metric. We established necessary and sufficient conditions for them to be MRD, MDS, AMDS, or NMDS, determined the conditions under which they are not MRD codes, and constructed several classes of MRD twisted Gabidulin codes. Moreover, we analyzed the covering radii of twisted Gabidulin codes and obtained two classes of deep holes for certain twisted Gabidulin codes. In summary, our work focused only on the case of twists at the same position. For the twisted Gabidulin codes with twists at different positions, their covering radii and deep holes remain to be further investigated. In addition, we expect more classes of MRD codes to be constructed.

%\section*{Ackonwledgements}

%National Natural Science Foundation of China under Grant No. 62201322
%This research was supported in part by the National Key Research and Development Program of China under Grant Nos. 2022YFA1004900, 2021YFA1001000 and 2022YFA1005000, the National Natural Science Foundation of China under Grant Nos. 62201322, 12141108, 62371259, 12471493 and 12441105,  the Natural Science Foundation of Shandong (ZR2022QA031), the Taishan Scholar Program of Shandong Province, the Natural Science Foundation of Tianjin (20JCZDJC00610), the Fundamental Research Funds for the Central Universities of China (Nankai University) and the Nankai Zhide Foundation.

\bibliographystyle{model1a-num-names}

\end{document}